\documentclass[a4paper,12pt]{amsart}
\usepackage{amsmath,amssymb,amsthm}
\usepackage[a4paper,margin=2cm]{geometry}
\usepackage{mathtools}
\usepackage{url}
\usepackage{color}

\usepackage[all]{xy}

\usepackage[colorlinks,linkcolor=blue, citecolor=blue]{hyperref}

\usepackage[backend=bibtex,style=alphabetic,url=true,doi=false,giveninits=true,isbn=false,maxbibnames=10]{biblatex}
\AtEveryBibitem{\clearlist{language}}
\newcommand*{\pd}
[2]{\mathchoice{\frac{\partial#1}{\partial#2}}
  {\partial#1/\partial#2}{\partial#1/\partial#2}
  {\partial#1/\partial#2}}
\newcommand*{\fd}
[2]{\mathchoice{\frac{\delta#1}{\delta#2}}
  {\delta#1/\delta#2}{\delta#1/\delta#2}
  {\delta#1/\delta#2}}

\newcommand*{\pdS}[2]{\D_{#2}#1}
\newcommand*{\tfdS}[2]{\tilde{\delta}_{#2}#1}
\newcommand*{\fdS}[2]{{\delta}_{#2}#1}

\theoremstyle{plain}
\newtheorem{theorem}{Theorem}[section]
\newtheorem{conjecture}[theorem]{Conjecture}
\newtheorem{corollary}[theorem]{Corollary}
\newtheorem{lemma}[theorem]{Lemma}

\newtheorem{proposition}[theorem]{Proposition}

\theoremstyle{definition}
\newtheorem{definition}[theorem]{Definition}

\theoremstyle{remark}
\newtheorem{remark}[theorem]{Remark}
\newtheorem{convention}[theorem]{Convention}

\newcommand{\D}{{\partial}}

\newcommand{\cA}{\mathcal{A}}
\newcommand{\cF}{\mathcal{F}}
\newcommand{\cS}{\mathcal{S}}
\newcommand{\cE}{\mathcal{E}}
\newcommand{\cG}{\mathcal{G}}
\newcommand{\cR}{\mathcal{R}}
\newcommand{\cP}{\mathcal{P}}

\DeclareMathOperator{\ad}{ad}

\allowdisplaybreaks[3]

\begin{document}

\title{Miura-reciprocal transformations and localizable Poisson pencils}
\author{P. Lorenzoni} \address[P. Lorenzoni]{Department of Mathematics and
  Applications, University of Milano Bicocca, Via Roberto Cozzi 55, 20125
  Milano, Italy and INFN sezione di Milano-Bicocca}
\email{paolo.lorenzoni@unimib.it}

\author{S. Shadrin} \address[S. Shadrin]{Korteweg--de Vries Institut for
  Mathematics, University of Amsterdam, Postbus 94248, 1090 GE Amsterdam, The
  Netherlands} \email{s.shadrin@uva.nl}

\author{R. Vitolo} \address[R. Vitolo]{Department of Mathematics and Physics
  \textquotedblleft E. De Giorgi\textquotedblright, University of Salento,
  via per Arnesano, 73100 Lecce, Italy and INFN sezione di Lecce}
\email{raffaele.vitolo@unisalento.it}

\date{}

\keywords{bi-Hamiltonian PDE, Hamiltonian operator, Miura transformation,
  reciprocal transformation,
  integrable systems.}

\subjclass[2020]{37K05, 37K10, 37K20, 37K25}
\maketitle

\begin{abstract} We show that the equivalence classes of deformations of
  localizable semisimple Poisson pencils of hydrodynamic type with respect to
  the action of the Miura-reciprocal group contain a local representative and
  are in one-to-one correspondence with the equivalence classes of deformations
  of local semisimple Poisson pencils of hydrodynamic type with respect to the
  action of the Miura group.
\end{abstract}

\tableofcontents

\section{Introduction}
\label{sec:preliminaries}

In 2001, Dubrovin and Zhang initiated a classification programme of
bi-Hamiltonian integrable PDEs in two independent variables
\cite{dubrovin01:_normal_pdes_froben_gromov_witten}. The group action that they
considered was that of Miura transformations, i.e., transformations depending
on the field variables and, polynomially, by their derivatives of higher order
through a perturbative series.

Among the questions that the above approach raises there is the issue of
extending the group action to include (possibly non-local) changes of variables
in one of the independent variables. Indeed, an important class of such
transformations is that of reciprocal transformations, which play an important
role in Mathematical Physics (see e.g. \cite{rogers69:_invar,rogers68:_recip,ferapontov89:_recip_trans,MR1127511,ferapontov03:_recip_hamil,MR2207050,MR2372106,MR2525536,BLASZAK2009341,LiuZhang:JacobiStructures,MR3054294}).

This paper is concerned with the action of the group of Miura-reciprocal
transformations, that is a natural group of simultaneous transformations of the
independent and the dependent variables of a (bi-)Hamiltonian system through a
perturbative series of derivatives of the field variables. Among other things,
we consider (1) the Miura-reciprocal transformations of the 1st kind and
rederive from the scratch the Ferapontov--Pavlov formula for the transformation
of a hydrodynamic bivector; (2) Miura-reciprocal transformations of the 2nd
kind (close to identity) and classify the orbits of their action on Poisson
pencils of weakly non-local bi-vectors of localizable shape with localizable
semi-simple hydrodynamic leading term; (3) a smaller group of
projective-reciprocal transformations and prove that they preserve the
Doyle--Pot\"emin canonical form of the bi-vectors.

A detailed comparison between previous results and our results can be read in the following
Subsections; we just stress that our result on the classification of bi-Hamiltonian integrable
structures provides a natural extension for the classification program in~\cite{dubrovin01:_normal_pdes_froben_gromov_witten} (that also incorporates and explains some of the results in~\cite{LiuZhang:JacobiStructures}). To the best of our knowledge it is the first result in the literature that provides a systematic classification of orbits of the action of the group of Miura-reciprocal transformations in the bi-Hamiltonian context (for a single Poisson structure this type of result is established in~\cite{ferapontov03:_recip_hamil,LiuZhang:JacobiStructures}).

\subsection{A variety of jet space transformation groups}

We consider a jet space $J^r(1,N)$, $r\geq 0$, with independent variable $x$
and dependent variables $u^i$, $i=1, \dots, N$, considered as coordinates on
some open domain $U\subset \mathbb{R}^N$. Let $u^{i,\sigma}$ denote the
$x$-derivative of $u^i$ taken $\sigma$ times.

Consider the transformations (i.e.~diffeomorphisms) of the jet space
$J^r(1,n)$. We begin from the most general type of transformation: a reciprocal
transformation coupled with a differential substitution. Reciprocal
transformations in a modern setting were introduced in
\cite{rogers69:_invar,rogers68:_recip} in the study of gas dynamics, and later
analyzed under a geometric viewpoint in
\cite{ferapontov89:_recip_trans,MR1127511} and
many other authors (see
e.~g.~\cite{ferapontov03:_recip_hamil,MR2207050,MR2372106,MR2525536,BLASZAK2009341,LiuZhang:JacobiStructures,MR3054294}
and references therein). The class of differential substitutions was introduced
in \cite{Ibragimov:TGApMP}, although many particular differential substitutions
were already present in the literature (in particular, the Miura
transformations).

\begin{definition}
  A \emph{reciprocal differential substitution} is a nonlocal transformation of
  the independent variable $x$ into the independent variable $y$ of the type
  \begin{equation}
    \label{eq:5}
    dy = Bdx,\qquad B=B(x,u^i,u^{i,\sigma})
  \end{equation}
  coupled with a differential substitution of the dependent variables of the
  form
  \begin{equation}
    \label{eq:2}
    w^i = Q^i(x,u^j,u^{j,\sigma}).
  \end{equation}
\end{definition}

By the fact that $dx(\partial_x)=1=dy(\partial_y)$ we obtain that total
derivatives are related by the formula
$\partial_x = B\partial_y$.
Note that, in general, the inversion of a differential substitution is a
nonlocal operation. We will soon focus on a more restrictive class of
transformations.

Reciprocal differential substitutions admit several interesting subclasses:

\begin{itemize}
\item A \emph{reciprocal transformation} is a nonlocal transformation of the
  independent variables $x$ into the independent variable $y$ of the type
  \begin{equation}
    \label{eq:5b}
    dy = Bdx,\qquad B=B(x,u^i,u^{i,\sigma})
  \end{equation}
  coupled with the identical transformation of the dependent variables. In
  practical applications the functions $u^i$ depend also on an additional
  parameter that plays the role of ``time'' of the system of evolutionary PDEs
  \[u^i_t=F^i(x,u^i,u^{i,\sigma}),\qquad i=1,...,n\] governing their
  evolution. Taking into account this additional variable reciprocal
  transformations are often defined as
  \[dy = Adt+Bdx,\] where the function $A,B$ are submitted to the closure
  condition $B_t=A_x$, 
  that is, $dy$ is a conservation low for the equation. Note that the
  coefficient $A$ doesn't enter the transformation law for $\partial_x$, and
  thus can be disregarded throughout this paper.

\item A reciprocal differential substitution is said to be \emph{holonomic} if
  there exists a differential function $P$ such that $B=\partial_x P$.
\item A general differential substitution of $(x,u^i)$ into $(y,w^j)$:
  \begin{equation}
    \label{eq:18}
    y=P(x,u^j,u^{j,\sigma}),\qquad  w^i=Q^i(x,u^j,u^{j,\sigma}),
  \end{equation}
  yields a holonomic reciprocal differential substitution $dy=\partial_x P dx$,
  $w^i=Q^i$ by differentiation (in this sense the two classes of
  transformations coincide).
\end{itemize}

The above two categories of transformations, basically local and nonlocal
differential substitutions, are, on the one hand, too wide to be used in the
context of the classification programs for evolutionary PDEs and related
geometric structures as, for instance, the one initiated by Dubrovin and Zhang
in \cite{dubrovin01:_normal_pdes_froben_gromov_witten}, and on the other hand
too restrictive since we are limited by fixing the parameter $r\geq 0$ that
controls the maximal order of jets.

For this reason, we introduce the space of differential polynomials
$\mathcal{A}$, and the following group of transformations, which is a subclass
of the reciprocal differential substitutions. Consider a jet space
$J^\infty(1,N)$ (considered as an inductive limit of the jet spaces $J^r(1,N)$,
$r\to \infty$) with independent variable $x$ and dependent variables $u^i$,
$i=1, \dots, N$.
Denote $u^{i,\sigma}$ the $x$-derivative of $u^i$ taken $\sigma$ times. We
associate with this space the algebra of functions
$\cA \coloneqq C^\infty(U)[[u^{i,\sigma}, i=1,\dots,N, \sigma\geq 1]]$, where
$C^\infty(U)$ is the space of smooth functions on a domain
$U\subset \mathbb{R}^N$ in the coordinates $u^1,\dots,u^N$.
There is a natural gradation on the algebra of densities $\cA$ given by
$\deg_{\D_x} u^{i,\sigma}$.
Let $\cA_d$ 
denote the $\deg_{\D_x}$-degree $d$ part of $\cA$, which is a finite
dimensional module over $C^\infty(U)$. 

\begin{definition} \label{def:Miura-reciprocal-trans} A Miura-type reciprocal
  differential substitution, or Miura-reciprocal transformation for short, is a
  transformation of the type
  \begin{align}
    \label{eq:756}
    dy &= \left(\sum^\infty_{k=0}
         \epsilon^kH_k(u^j,u^{j,1},\ldots,u^{j,k})\right)dx,
    \\ \notag
    w^i & = \sum^\infty_{k=0}
          \epsilon^kK^i_k(u^j,u^{j,1},\ldots,u^{j,k}),\quad i=1,\ldots,N,
  \end{align}
  with $H_k,K^i_k\in\mathcal{A}_k$ and
  \begin{align}
    H_0\neq 0, \qquad \det\left(\frac{\partial K^i_0(u^j)}{\partial u^k}
    \right)\neq 0.
  \end{align}
\end{definition}
The formal dispersive parameter $\epsilon$ that we introduce here to control
the $\deg_{\D_x}$-degree is, in principle, not strictly necessary but it is
very convenient in particular computations and applications.

The set of all Miura-reciprocal transformations is denoted by $\mathcal{R}$. It
is a group with respect to the composition, and it has some distinguished
subgroups:
\begin{itemize}
\item the subgroup $\mathcal{R}_{DS}$ of Miura differential substitutions, that
  are Miura-type reciprocal differential substitutions which are also holonomic
  differential substitutions of the following type:

  There exists $P=x+P_0$, with $P_0=\sum_{k=0}^\infty\epsilon^kF_k$ and
  $F_k\in\mathcal{A}_k$, such that
  \begin{equation}
    \label{eq:80}
    \partial_xP = \sum^\infty_{k=0}
    \epsilon^kH_k(u^j,u^j_x,\ldots,u^j_\sigma);
  \end{equation}
\item the subgroup of Miura transformations characterized by $H_0=1$ and
  $H_{k}=0$ for all $k\geq 1$. This subgroup is called the \emph{Miura group}
  $\mathcal{G}\subset\mathcal{R}$
  \cite{dubrovin01:_normal_pdes_froben_gromov_witten} and bears his name from
  the transformation relating KdV and modified KdV equations introduced by Miura
  \cite{Miura1968KortewegdeVE}. Note that the Miura group is also a subgroup of
  the group of Miura differential substitutions:
  $\mathcal{G}\subset\mathcal{R}_{DS}$.
\end{itemize}

\begin{definition}
  By analogy with the way the standard Miura group is typically presented, we
  introduce the following two subgroups.
  \begin{itemize}
  \item We define Miura-reciprocal transformations of the \emph{1st kind} to be
    the Miura-reciprocal transformations of the form
    \begin{align}
      \label{eq:Miura-reciprocal-1st-kind}
      dy & = H_0(u^j)dx,
      \\ \notag
      w^i & =K^i_0(u^j), && i=1,\ldots,N.
    \end{align}
    The group of all Miura RDS of the first type is denoted by $\mathcal{R}_I$.
    This group contains as a subgroup the group of Miura transformations of the
    1st kind, $\mathcal{G}_I\subset\mathcal{R}_I$, characterized by $H_0=1$.
  \item We define Miura-reciprocal transformations of the \emph{2nd kind} to be
    the Miura-reciprocal transformations of the form
    \begin{align}
      \label{eq:Miura-reciprocal-2nd-kind}
      dy & = \left(1 + \sum^\infty_{k=1}
           \epsilon^kH_k(u^j,u^{j,1},\ldots,u^{j,\sigma})\right)dx,
      \\ \notag 
      w^i & = u^i + \sum^\infty_{k=1}
            \epsilon^kK^i_k(u^j,u^{j,x},\ldots,u^{j,\sigma}),&&  i=1,\ldots,N.
    \end{align}
    The group of all Miura-reciprocal transformations of the second type is
    denoted by $\mathcal{R}_{II}$. It contains as a subgroup the group of Miura
    transformations of the 2nd kind, $\mathcal{G}_{II}\subset\mathcal{R}_{II}$,
    characterized by $H_k=0$ for all $k\geq 1$.
  \end{itemize}
\end{definition}

\begin{definition} A distinguished subgroup of $\mathcal{R}_I$ is the group of
  \emph{projective reciprocal transformations} $\cP$. Such transformations are
  characterized by the requirements that $K^i$ in
  Equation~\eqref{eq:Miura-reciprocal-1st-kind} is a projective transformation
  (in an affine chart) and $H_0$ is the common denominator of the projective
  transformation. More explicitly,
  \begin{align}
    \label{eq:81}
    dy & = (a^0_j u^j + a^0_0)dx,\\ \notag 
    w^i & =\frac{a^i_j u^j + a^i_0}{a^0_ju^j + a^0_0}, &&  i=1,\ldots,N.
  \end{align}
\end{definition}

The goal of this paper is to discuss some aspects of the actions of these
groups on the natural suitable geometric structures that emerge in the study of
integrable systems of evolutionary PDEs. In order to describe our results we
have to recall some of these structures, which we do in the rest of the
introduction.

\subsection{Action of the transformation groups}
\label{sec:acti-transf-groups}

The above group of Miura reciprocal differential substitutions act on spaces of
geometric entities that play important roles in the geometric theory of
integrability. In particular, it acts on:
\begin{itemize}
\item densities, that have the form
  \begin{equation}
    \label{eq:82}
    F=\int f(u^j,u^{j,\sigma})\,dx\in \cF \coloneqq \mathcal{A}/\partial_x\mathcal{A},
    \quad\text{with}\quad f\in\mathcal{A};
  \end{equation}
\item variational vector fields, that include symmetries of partial
  differential equations, and have the form
  \begin{equation}\label{eq:83}
    \varphi=\varphi^i(u^j,u^{j,\sigma})\fdS{}{u^i},\qquad
    \varphi^i\in\mathcal{A};
  \end{equation}
\item covector-valued densities, that include characteristic vectors of
  conserved quantities of differential equations, and have the form
  \begin{equation}
    \label{eq:84}
    \psi=\psi_i(u^j,u^{j,\sigma})du^i\otimes dx,\qquad \psi_i\in\mathcal{A};
  \end{equation}

\item the Euler--Lagrange operator, which sends densities into covector-valued
  densities,
  \begin{equation}
    \label{eq:85}
    \mathcal{E}(F)=\fdS{F}{u^i}du^i\otimes dx;
  \end{equation}
\item variational multivectors of degree $p$, that include Hamiltonian
  operators of partial differential equations as particular bivectors. They can
  be regarded as maps from $(p-1)$-covector-valued densities to variational
  vector fields.
\end{itemize}

In Section~\ref{sec:diff-subst} we will prove our change of coordinate formulae
for reciprocal differential substitutions
for these geometric objects.  As an example, we re-derive in
Section~\ref{sec:F-P} the Ferapontov--Pavlov formula for the reciprocal
transformation of a Poisson bi-vector of the differential order $1$
\cite{ferapontov03:_recip_hamil}, and this brings us to the realm of weakly
non-local Poisson structures of localizable shape.

\subsection{Weakly non-local Poisson bi-vectors of localizable shape} Let
dependent variables $u^i$ also dependent on one external parameter, denoted by
$t$. The most studied structures in geometric theory of integrability are the
local Poisson structures needed for representation of equations of the form
\begin{align}
  u^i_t = f^i(u^j,u^{j,\sigma})
\end{align}
in Hamiltonian form (note that we don't allow possible explicit dependence of
$f^i$'s on $x$), that is, in the form
\begin{align}
  u^i_t = \sum_{s=0}^d P^{ij}_s \D_s \fd{}{u^j} \int h(u^k,u^{k,\sigma}) dx,
\end{align}
where $H = \int h(u^k,u^{k,\sigma}) dx$ is the Hamiltonian functional and
$P=\sum_{s=0}^d P^{ij}_s \D_s$, $P^{ij}_s\in \cA$ defines a bi-vector which in
the language of densities can be written as
\begin{align}
  \{u^i(x),u^j(y)\}_P = \sum_{s=0}^d P^{ij}_s \D_x^s \delta(x-y)
\end{align}
(in this paper bi-vectors and, more generally, multivectors are assumed to be
skew-symmetric by default).

In addition to the language of densities, there is a very convenient formalism,
the so-called $\theta$-formalism, to encode the variational
multivectors~\cite{Getzler:DTHOpFCV}, see also~\cite{IVV}. Namely, extend the space $\cA$ to a
space $\hat\cA\coloneqq \cA[[\theta_i^\sigma, i=1,\dots,N, \sigma\geq 0]]$,
where $\theta_i^\sigma$ are formal odd variables. We often denote $\theta_i^0$
by $\theta_i$, and we extend the $\D_x$ operator to $\hat\cA$ as
$\D_x \coloneqq \sum_{s=0}^\infty u^{i,s+1}\D_{u^{i,s}} +
\theta_i^{s+1}\D_{\theta_i^s}$.  The $\deg_{\D_x}$-gradation is extended to
$\hat\cA$ by $\deg_{\D_x} \theta_i^\sigma = \sigma$, and there is a natural
$\theta$-degree given by $\deg_{\theta} u^{i,\sigma}=0$ and
$\deg_{\theta} \theta_i^{\sigma}=1$. Let $\hat\cA^p$ denote the subspace
$\hat\cA$ of $\theta$-degree $p$. We can consider it as a space of densities of
variational $p$-vectors. Let $\hat\cA_d^p \coloneqq \hat\cA_d \cap \hat\cA^p$.

The space $\hat\cF \coloneqq \hat\cA/\D_x \hat\cA$ can be considered as the
space of variational multivectors. It inherits under the projection
$\int \colon \hat\cA\to \hat\cF$ both gradations, $\deg_{\D_x}$ and
$\deg_{\theta}$, and $\cF^p_d$ denotes its subspace of $p$-vectors of
differential degree $\deg_{\D_x}=d$.  The Schouten bracket is defined as
\begin{equation}
  \Big[\int P,\int Q\Big]=\int (-1)^{\deg_\theta P}\delta_{u^i}{P}\delta_{\theta_i}{Q} +
  \delta_{\theta_i}{P}\delta_{u^i}{Q}
\end{equation}
for $P,Q\in \hat\cA$, where
$\delta_{u^i}\coloneqq \sum_{s=0}^\infty (-\partial_x)^s\partial_{u^{i,s}}$ and
$\delta_{\theta_i}\coloneqq \sum_{s=0}^\infty
(-\partial_x)^s\partial_{\theta_i^{s}}$.  Various cohomological computations in
terms of this space and related formalism allow to efficiently control the
deformation theory of Poisson bi-vectors and their pencils, see
e.~g. \cite{LiuZhang-deformss,LiuZhang:JacobiStructures,DubrovinLiuZhang:QuasiTriv,LiuZhang:biham,CarletPosthumaShadrin:DeformSS,CKS,CPS-1,CPS-2,CCS-1,CCS-2}.

However, studying the action of the group of Miura-reciprocal transformations
we can not restrict ourselves to the local Poisson bi-vectors. As we have seen,
the reciprocal transformations generate non-locality of some very particular
shape, and in terms of the operator $P$ we have to extend its possible shape to
\begin{align}
  P = \sum_{s=0}^d P^{ij}_s \D_s + u^{i,1} \D_x^{-1} V^j + V^i \D_x^{-1} u^{j,1}, \qquad P^{ij}_s, V^i\in \cA.
\end{align}
Hamiltonian operators of the form above with $d=1$, $P^{ij}_1=g^{ij}(u)$
(det$\,g^{ij}\ne 0$), $P^{ij}_0=-g^{il}\Gamma^j_{lk}u^k_x$ and
$V^i=V^i_j(u)u^j_x$ were studied by Ferapontov in
\cite{ferapontov-ConformallyFlatMetrics95}. They belong to the larger class of
\emph{weakly non-local operators}, that was introduced in \cite{MR1855607}.
Like in the local case the coefficients $g^{ij}$ define a metric and the
coefficients $\Gamma^j_{lk}$ the Christoffel symbols of the associated
Levi-Civita connection but unlike in the local case the metric is no longer
flat. It turns out that the Riemann tensor $R$ and the tensor field $V$
defining the non-local tail satisfy the conditions
\[g_{is}V^s_j=g_{js}V^s_i,\quad\nabla_jV^k_i=\nabla_iV^k_j,\quad
  R^{ij}_{kl}=V^i_k\delta^j_l+V^j_l\delta^i_k-V^j_k\delta^i_l-V^i_l\delta^j_k.\]
These are a particular instance of Ferapontov's conditions for weakly non-local
Hamiltonian operators of hydrodynamic type \cite{F95:_nl_ho}. An algorithm to
compute such conditions for general weakly non-local Hamiltonian operators has
been developed in \cite{https://doi.org/10.1111/sapm.12302} and implemented in
three different computer algebra systems in \cite{Casati_2022}.

A natural question here is how to extend the $\theta$-formalism briefly
recalled above to accommodate this type of non-locality. There are two recipes
in the literature given in~\cite{LiuZhang:JacobiStructures} (specific for this
case) and~\cite{LORENZONI2020103573} (suitable for general weakly non-local
operators). The identification of the two approaches should indirectly follow
from the uniqueness arguments in~\cite{LiuZhang:JacobiStructures}, but we
wanted to establish an explicit identification. We do it by an explicit
computation in Section~\ref{sec:LV-LZ}.

\begin{remark} It is important to comment on the action of the operator
  $\partial_x^{-1}$. It can be defined on $\partial_x\cA$ by
  $\partial_x^{-1} (\partial_x (f)) = f +C$ for any $f\in \cA$, here $C$ is
  some constant.  For a more general element $g\in \cA$,
  $g\not\in\partial_x\cA$, we can represent $\partial_x^{-1} (g)$, for
  instance, as an element of a localization of $\cA$ given by
  $\cA((\tfrac1{u^{1,1}}))$, that is, we can perturbatively represent it as a
  series $C+\sum_{i=1}^{\infty} \tfrac{h_i}{(u^{1,1})^i}$ with $h_i\in \cA$
  such that $\partial_{u^{1,1}} h_i = 0$ (this idea is coming
  from~\cite{DubrovinLiuZhang:QuasiTriv}), here $C$ is also a constant.

  Both approaches that we compare assert that for the analysis of the weakly
  non-local Poisson bi-vectors of localizable shape it is sufficient to
  formally apply $\partial_x^{-1}$ to just one element
  $-u^{i,1}\theta_i\in\hat{\cA}$ and denote the result by $\zeta$, which has
  different meaning in these two approaches. The subsequent usage of $\zeta$ in
  computations implies that the extra constant that might occur by inverting
  $\partial_x$ is uniformly set to $C=0$.
\end{remark}

\subsection{Localizability} Consider a dispersive weakly non-local Poisson
structure of localizable shape given by
\begin{align}
  P = \sum_{d=1}^\infty \epsilon^{d-1} \left(\sum_{s=0}^d P^{ij}_{d,d-s} \D_s + u^{i,1} \D_x^{-1} V_d^j + V_d^i \D_x^{-1} u^{j,1}\right), \qquad P^{ij}_{d,k}, V_k^i\in \cA_k.
\end{align}
The leading term ($d=1$) of this structure is a Poisson structure of
hydrodynamic type and thus the full Poisson structure can be thought as a
deformation of a Poisson structure of hydrodynamic type.  If
$\det P^{ij}_{1,0}\not=0$, Liu and Zhang prove
in~\cite{LiuZhang:JacobiStructures} that there is always an element in $\cR$
that turns $P$ into a constant local Poisson structure $\eta^{ij}\D_x$. In the
case of a purely local structure the same results is established under the
action of group $\cG$ in~\cite{Getzler:DTHOpFCV} (see
also~\cite{degiovanni05:_poiss}
and~\cite{dubrovin01:_normal_pdes_froben_gromov_witten}), and in the case
$\epsilon =0$ (that is, a purely degree $1$ case) it is established under the
action of the group $\cR_I$ in~\cite{LiuZhang:JacobiStructures} for $N=1,2$ and
in~\cite{ferapontov03:_recip_hamil} for $N\geq 3$.

Now consider a pencil $P-\lambda Q$ of dispersive weakly non-local Poisson
structure of localizable shape. Let us fix the leading term
$(P-\lambda Q)|_{\epsilon = 0}$ of the pencil and assume it is semi-simple. In
the purely local case (that is, under the additional assumption that both $P$
and $Q$ are purely local), it was suggested
in~\cite{LiuZhang-deformss,DubrovinLiuZhang:QuasiTriv} (see
also~\cite{LORENZONI2002331} for the scalar case) and proved
in~\cite{LiuZhang:biham} ($N=1$ case)
and~\cite{CarletPosthumaShadrin:DeformSS,CKS} (any $N\geq 1$) that the space of
orbits of the action of $\cG_{II}$ on such pencils is naturally parametrized by
$N$ smooth functions of one variable, called the \emph{central invariants}.

In Section~\ref{sec:pencils} we generalize these results in the following
way. Let us fix the leading term $(P-\lambda Q)|_{\epsilon = 0}$ of the pencil
and assume that $P|_{\epsilon = 0}$ and $Q|_{\epsilon = 0}$ are simultaneously
localizable under the action of the group $\cR_I$. We also still assume that
$(P-\lambda Q)|_{\epsilon = 0}$ is semi-simple. In this case, we prove that the
set of orbits of the action of $\cR_{II}$ on such pencils is also naturally
parametrized by $N$ smooth functions of one variable. Note that while the
statement is literally the same as in the purely local case, it is quite
different as both the group and the space of structures on which the group acts
is much bigger. We show that it is still possible to read the central
invariants from the symbol of the pencil.

This result is proved by a direct application of various techniques and results
proposed
in~\cite{LiuZhang:JacobiStructures,LiuZhang:biham,CarletPosthumaShadrin:DeformSS,CKS}. From
the comparison with the computations in the local case, we obtain the following
extra result: under the assumptions above, each orbit of $\cR_{II}$ contains a
purely local representative. In other words, we prove that if the leading term
of a semi-simple pencil $P-\lambda Q$ of dispersive weakly non-local Poisson
structure of localizable shape is localizable by the action of the group
$\cR_I$, then the whole pencil is localizable by the action of the group $\cR$.

It is worth to mention that this result also generalizes and put in the right
context a theorem of Liu and Zhang~\cite[Theorem
1.3]{LiuZhang:JacobiStructures} that states that if two local Poisson pencils
with the leading semi-simple hydrodynamic term are related by a reciprocal
transformation, then their central invariants are the same.

\subsection{Projective group and Doyle--Pot\"emin form}
\label{sec:proj-group-doyle}

Finally, we consider the action of the group $\cP\subset\cR_I$. It is a quite
small group with transparent structure, and we expect that in general the
orbits of its action should have a rich geometric structure. In this paper we
find a surprising connection of this group to a conjecture of Mokhov on the
possible form of the local Poisson structures of differential degree
$\deg_{\D_x}\geq 2$.

It was independently proved by Doyle \cite{MR1210220} and Pot\"emin
\cite{potemin91:_PhDThesis,MR1479400} that homogeneous local Poisson structures
of differential degree $d=2$, $3$, i.e. of the form
\begin{equation}
  \sum_{s=0}^d P^{ij}_{d-s} \D_x^s, \qquad P^{ij}_{k}\in \cA_k,
\end{equation}
can always be transformed by the action of the group $\cG_I$ to an operator of
the shape
\begin{equation}\label{eq:1}
  \D_x \circ \sum_{s=0}^{d-2}Q^{ij}_{d-2-s} \D_x^s \circ \D_x,
  \qquad Q^{ij}_{k}\in \cA_k.
\end{equation}
Mokhov made the following interesting conjecture:

\begin{conjecture}[{See~\cite[Proposition 2.3 and text
    afterwards]{MokhovSurvey1998}}]  Let $P=\sum_{e=1}^{d+2} P^{ij}_e
  \D_x^{d+2-e}$ be a local operator of homogeneous differential order $d+2$
  (that is, $\deg_{\D_x} P^{ij}_e=e$), $d\geq 0$. Assume that $P$ defines a
  Poisson bracket. Then there exists a local skew-symmetric operator  $Q^{ij}$
  of homogenenous differential order $d$ such that $P =\D_x \circ Q^{ij} \circ
  \D_x$.
\end{conjecture}

The form \eqref{eq:1} is called the
Doyle--Pot\"emin form of a local homogeneous bi-vector of differential degree
$\deg_{\D_x}\geq 2$.

It was recently proved that in the cases of homogeneous local Poisson
structures of degree $d=2$ \cite{vv:_projec_hamil} and $d=3$ \cite{MR3253545}
the form~\eqref{eq:1} is preserved by the action of the group $\cP$.  In
Section~\ref{sec:Projective}, thanks to our change of coordinates formulae from
Subsection~\ref{sec:acti-transf-groups}, we generalize the above results to
local homogeneous bi-vectors (i.e., not necessarily Poisson structures) of
degree $d\geq 2$ and prove that the group $\cP$ preserves the set of local
bi-vectors of Doyle--Pot\"emin form. A nice example of application to a
Hamiltonian operator for the Dubrovin--Zhang hierarchy is pointed out.

\subsection{Acknowledgments}
S.~S. and R.~V. were supported by the Netherlands Organization for Scientific
Research. P.~L. and R.~V. are supported by funds of INFN (Istituto Nazionale di
Fisica Nucleare) by IS-CSN4 Mathematical Methods of Nonlinear Physics. P.~L. is
supported by funds of H2020-MSCA-RISE-2017 Project No. 778010
IPaDEGAN. P.~L. and R.~V. are also thankful to GNFM (Gruppo Nazionale di Fisica
Matematica) for supporting activities that contributed to the research reported
in this paper.

\section{Formulae for the action}
\label{sec:diff-subst}

The goal of this Section is to compute from the scratch the effect of general
reciprocal differential substitutions on variational (multi)vector fields and
related geometric objects.  It is clear that a reciprocal differential
substitution given by $dy=Bdx$ (or $y=P$ in the holonomic case), $w^i = Q^i$,
also yields a coordinate change of the $y$-derivative variables:
\begin{equation}
  \label{eq:6}
  w^{i,\tau} = \partial_y^{\tau} Q^i =
  \left(\frac{1}{B}\partial_x\right)^{\tau}Q^i.
\end{equation}

It is convenient to introduce the Fr\'echet derivative\footnote{It should be
  the Gateaux derivative, but Fr\'echet is prevailing in the literature.} of a
differential function $F\in\cA$, as
\begin{equation}
  \label{eq:3}
  \ell_F(X) = (\ell_F)_i (X^i) = \sum_{\sigma=0}^\infty \pd{F}{u^{i,\sigma}}\partial_x^\sigma X^i,
\end{equation}
where $X=X^i\delta_{u^i}$ is a variational vector field.  The formal adjoint of
the above operator is
\begin{equation}
  \label{eq:4}
  (\ell_F^*)_i = \sum_{\sigma=0}^\infty (-\partial_x)^\sigma\circ
  \pd{F}{u^{i,\sigma}}
\end{equation}
acting on covector-valued densities.

A change of coordinates formula for Hamiltonian operators under the action of
differential substitutions was already given in
\cite{MR911774,olver88:_darboux_hamil}.  We rephrase the arguments of the proof
in \cite{olver88:_darboux_hamil} and obtain change of coordinates formulae for
the geometric objects that we listed in Subsection~\ref{sec:acti-transf-groups}
which turn out to be valid in the more general case of reciprocal differential
substitutions.

We observe that also in \cite{LiuZhang:JacobiStructures} there are formulae for
coordinate change, but their validity is limited to the action of Miura
reciprocal transformations on operators of localizable shape, while we do not
have this limitation.

First of all, we provide a formula for the coordinate change of an variational
vector field under a differential substitution.
\begin{proposition}\label{pro:evfield}
  Let $X^i\delta_{u^i}=Y^i\delta_{w^i}$ be a variational vector field in the
  coordinate systems $(x,u^{i,\sigma})$ and $(y,w^{i,\sigma})$, respectively,
  where the latter coordinates systems are related by a holonomic reciprocal
  differential substitution $y=P$, $w^i = Q^i$.  Then the following change of
  coordinate formula holds:
  \begin{equation}\label{eq:21}
    Y^j = \frac{1}{\partial_xP}\mathcal{D}^j_i (X^i)
  \end{equation}
  where
  \begin{equation}
    \label{eq:87}
    \mathcal{D}^j _i= \partial_xP(\ell_{Q^j})_i - \partial_x Q^j  (\ell_P)_i.
  \end{equation}
\end{proposition}

\begin{proof}
  The proof uses arguments that provide a change of coordinates formula for
  Euler--Lagrange operators in \cite{Olver:ApLGDEq}, Theorem~4.8 and
  Exercise~5.49. Let
  \begin{equation}
    \label{eq:40}
    u^i=f^i(x),\quad x\in\Omega,
    \qquad w^i=g^i(y),
    \quad y\in\tilde{\Omega}
  \end{equation}
  be functions that are put in correspondence by a transformation. We can
  consider a one-parameter family of such functions defined by the variation
  field $X^i\partial_{u^i}$:
  \begin{equation}
    \label{eq:41}
    u^i_\epsilon=f^i(x,\epsilon) = f^i(x) + \epsilon X^i(x),
  \end{equation}
  where $X^i\partial_{u^i}$ has compact support in $\Omega$.  Its transformed
  version
  \begin{equation}
    \label{eq:42}
    w^i_\epsilon=g^i(y,\epsilon) =
    g^i(y) + \epsilon Y^i(y) + \mathcal{O}(\epsilon^2).
  \end{equation}
  is determined by the formulae
  \begin{equation}
    \label{eq:42b}
    y=P(x,\partial_x^\sigma(f^j(x)+\epsilon X^j(x))),\quad
    w^i_\epsilon = Q^i(x,\partial_x^\sigma(f^j(x)+\epsilon X^j(x))).
  \end{equation}
  Since $\eta$ has compact support on $\Omega$, each $g^i(y,\epsilon)$ is
  defined on a common compact domain
  $\tilde{\Omega}=\{x\in\Omega \mid y=P(x,\partial_x^\sigma f^j(x))\}$.  The
  transformed variation field is given by
  $ Y^i(y) = \partial_\epsilon {g^i(y,\epsilon)}\big|_{\epsilon=0}$. As
  variation fields do not depend on $\epsilon$ we have
  \begin{equation}
    \label{eq:12}
    \partial_\epsilon y=0=
    \partial_x P \partial_\epsilon x  + \sum_{\sigma=0}^\infty
    \partial_{u^{j,\sigma}} P\partial_x^\sigma X^j ,
  \end{equation}
  hence
  \begin{equation}
    \label{eq:13}
    \partial_\epsilon x \big|_{\epsilon=0} =
    -\frac{1}{\partial_x P} \sum_{\sigma=0}^\infty
    \partial_{u^{j,\sigma}} P\partial_x^\sigma X^j.
  \end{equation}
  We have:
  \begin{align}
    \label{eq:14}
    Y^j & = \partial_\epsilon g^j (y,\epsilon)
          \big|_{\epsilon=0} =
          \sum_{\sigma=0}^\infty \partial_{u^{i,\sigma}} Q^j \partial_x^\sigma \partial_\epsilon f^i(x,\epsilon) \big|_{\epsilon=0}
          + \partial_x Q^j \partial_\epsilon x \big|_{\epsilon=0} 
    \\ \notag &
                =\frac{1}{\partial_xP}\left(\partial_x P (\ell_{Q^j})_i
                - \partial_x Q^j(\ell_P)_i\right)X^i.
  \end{align}
\end{proof}

In the non-holonomic case, we have to regard the differential function $P$ as
the primitive of a differential function $B$, $P=\partial_x^{-1}B$, and we
obtain the following Corollary.
\begin{corollary} In the non-holonomic case of the reciprocal differential
  substitution $dy = Bdx$, $w^i = Q^i$ the following change of coordinate
  formula holds for an variational vector field
  $X^i\delta_{u^i}=Y^i\delta_{w^i}$:
  \begin{equation}
    Y^j = \frac{1}{B}\mathcal{D}^j_i(X^i),\label{eq:10}
  \end{equation}
  where
  \begin{equation}
    \label{eq:19}
    \mathcal{D}^j_i = B(\ell_{Q^j})_i - \partial_x Q^j \circ \partial_x^{-1}\circ(\ell_B)_i.
  \end{equation}
\end{corollary}
Note that we used the property $\ell_B\circ
\partial^{-1}=\partial^{-1}\circ\ell_B$, which is very useful in computations.

Dualizing the computation above, we obtain the formulae for the change of
coordinates formula for the Euler--Lagrange operator.
\begin{corollary}\label{sec:diff-subst-1}
  Let the coordinate systems $(x,u^{i,\sigma})$ and $(y,w^{i,\sigma})$,
  respectively, where the latter coordinates systems are related by a
  reciprocal differential distribution $dy=Bdx$, $w^i = Q^i$, and let
  $\cE^x_i$, $\cE^y_i$ be the Euler--Lagrange operator with respect to the
  coordinates $(x,u^i_\sigma)$, $(y,w^{i,\sigma})$.  Then the change of
  coordinate formula is
  \begin{equation}
    \label{eq:50}
    \mathcal{E}_i^y=(\mathcal{D}^*)_i^k\circ{\mathcal{E}}_k^x,
  \end{equation}
  with $\mathcal D$ given by Equation~\eqref{eq:19}.
  \end{corollary}
  In the holonomic case, the formula reduces to the known formula
  in~\cite[Exercise 5.49]{Olver:ApLGDEq} (with $\mathcal D$ given by
  Equation~\eqref{eq:87}). In the particular case of a differential
  substitution of the dependent variable only we have $\ell_B=0$ and the
  above formula reduces to the well-known formula
  $\mathcal{E}^x=(\ell^*_{Q^k})_i \circ\mathcal{E}_k^y$.

\begin{corollary} Consider a covector-valued density
  $\Xi_i du^i \otimes dx = \Psi_i dw^i \otimes dy$ in the coordinate systems
  $(x,u^{i,\sigma})$ and $(y,w^{i,\sigma})$ related by $dy=Bdx$, $w^i = Q^i$.
  Then we have the following change of coordinates formula:
  \begin{equation}
    \label{eq:60}
    \Xi_i = (\mathcal{D}^*)_i^k (\Psi_k),
  \end{equation}
  where $\mathcal{D}$ is as in Equation~\eqref{eq:19} (or as in
  Equation~\eqref{eq:87} in the holonomic case $y=P$).
\end{corollary}

Finally, we obtain the following:

\begin{proposition} \label{thm:Miura-Reciprocal-Bivector} Consider a reciprocal
  differential substitution $dy=Bdx$, $w^i = Q^i$ and let $P^{ij}_x$,
  $P^{ij}_y$ be its coordinate expressions of a (possibly non-local or
  non-Poisson) bi-vector with respect to the coordinates $(x,u^i_\sigma)$ and
  $(y,w^i_\sigma)$. Then we have the change of coordinate formula
  \begin{equation}
    \label{eq:56}
    P^{hk}_y = \frac{1}{B}(\mathcal{D})^h_i P^{ij}_x
    (\mathcal{D}^*)_j^k,
  \end{equation}
  where $\mathcal{D}$ is as in Equation~\eqref{eq:19}.
\end{proposition}

\begin{proof}
  The proposition has already been proved in
  \cite{MR911774,olver88:_darboux_hamil} for the particular case of Hamiltonian
  operators and differential substitutions. In our case, the proof follows from
  the fact that $P^{ij}$ maps covector-valued densities into variational vector
  fields. So, we can use the change of coordinates formulae for these two
  geometric objects (independently of the Hamiltonian property) and find the
  above result, that holds also in the case of (nonlocal) reciprocal
  differential substitutions.
\end{proof}

\begin{remark} The same argument can be applied to multivector fields
  considered as maps from multicovector-valued densities to variational vector
  fields. For instance, in the same set-up as
  Theorem~\ref{thm:Miura-Reciprocal-Bivector} let ${T}^{ijk}$ and
  $\tilde{T}^{ijk} $ be the coordinate expressions of a trivector. Then
  \begin{equation}
    \label{eq:63}
    \tilde{T}^{ijk} = \frac{1}{B}(\mathcal{D})^i_m
    T^{mnp}((\mathcal{D}^*)_n^j,(\mathcal{D}^*)_p^k).
  \end{equation}
\end{remark}

\subsection{The Ferapontov--Pavlov formula}
\label{sec:F-P} Let us apply a
special case of Theorem~\ref{thm:Miura-Reciprocal-Bivector} to a local Poisson
bi-vector of order $\deg_{\D_x}=1$ and a reciprocal transformation in $\cR$
that only changes the independent variable. This should reproduce the
Ferapontov--Pavlov formula first derived in~\cite[Section
3]{ferapontov03:_recip_hamil} (based
on~\cite{ferapontov-ConformallyFlatMetrics95}).

Consider the change of $x$ given by
\begin{align} \label{eq:change-x} \D_x=B\D_y, \qquad \D_y ^{-1} B^{-1} =
  \D_x^{-1} \
\end{align}
as an element of $\cR_I$, that is, we assume that $B=B(u^j)$. Let a local
Poisson bracket of differential degree $1$ be given by the operator
\begin{align} \label{eq:operator-P} P^{ij}\coloneqq g^{ij} \D_x + \Gamma^{ij}_k
  u^k_x, \qquad (P^*)^{ji} = -P^{ij}
\end{align}

\begin{convention} \label{conv:ux} Throughout the computations in this Section
  it is important for us to distinguish between $\D_x u^k$ and $\D_y u^k$, so
  we use the notation $u^k_x$ and $u^k_y$ rather than $u^{k,1}$.
\end{convention}

\begin{proposition} \label{prop:F-P} The action of the reciprocal
  transformation~\eqref{eq:change-x} on the operator~\eqref{eq:operator-P}
  produces a weakly non-local operator of localizable shape, whose local part
  is given explicitly as
  \begin{equation} \label{eq:F-P-loc} g^{ij}B^2 \D_y + \Gamma^{ij}_k B^2 u^k_y
    -\frac 12 g^{i\ell} B^2 \left( g_{\ell m}\frac{\D B^{-2}}{\D u^k} + g_{km}
      \frac{\D B^{-2}}{\D u^\ell} -g_{\ell k} \frac{\D B^{-2}}{\D u^m}\right)
    g^{mj} B^2 u^k_y
  \end{equation}
  and the non-local part is equal to
  \begin{equation} \label{eq:F-P-nonloc} \left(P^{i\ell} \left(\frac{\D B}{\D
          u^\ell}\right) -\frac{1}{2} u^i_y \frac{\D B}{\D u^k} g^{k\ell}
      \frac{\D B}{\D u^\ell} \right) \D_y^{-1} u^j_y + u^i_y \D_y^{-1} \left(
      P^{jk}\left(\frac{\D B}{\D u^k}\right) -\frac{1}{2}\frac{\D B}{\D u^k}
      g^{k\ell} \frac{\D B}{\D u^\ell} u^j_y \right).
  \end{equation}
\end{proposition}

\begin{remark} Note that
  $\Gamma^{ij}_k B^2 -\frac 12 g^{i\ell} B^2 \left( g_{\ell m}\pdS{
      B^{-2}}{u^k} + g_{km} \pdS{ B^{-2}}{ u^\ell} -g_{\ell k} \pdS{ B^{-2}}{
      u^m}\right) g^{mj} B^2$ is exactly the covariant Christoffel symbol for
  the metric $g^{ij} B^2$, so we indeed reproduce the Ferapontov--Pavlov
  formula in~\cite{ferapontov03:_recip_hamil}.
\end{remark}

\begin{proof}[Proof of Proposition~\ref{prop:F-P}]
  By Theorem~\ref{thm:Miura-Reciprocal-Bivector} the bi-vector $P^{ij}$ is
  transformed under the substitution~\eqref{eq:change-x} to
  \begin{align}
    B^{-1} \left(B\delta^i_k - u^i_x \D_x^{-1} \frac{\D B}{\D u^k} \right) P^{k\ell} 
    \left(B\delta^j_\ell + \frac{\D B}{\D u^\ell} \D_x^{-1} u^j_x  \right).
  \end{align}
  Expanding the brackets, we have the following four summands (we intentionally
  keep derivatives in $x$ instead of $y$ as long as possible):
  \begin{align}
    -B^{-1} u^i_x \D_x^{-1} \frac{\D B}{\D u^k} P^{k\ell} \frac{\D B}{\D u^\ell} \D_x^{-1} u^j_x  & = -\frac{1}{2} B^{-1} u^i_x \D_x^{-1} \left( \D_x \frac{\D B}{\D u^k} g^{k\ell} \frac{\D B}{\D u^\ell} + \frac{\D B}{\D u^k} g^{k\ell} \frac{\D B}{\D u^\ell} \D_x\right) \D_x^{-1} u^j_x \\ \notag 
                                                                                                  &=
                                                                                                    -\frac{1}{2} B^{-1} u^i_x \frac{\D B}{\D u^k} g^{k\ell} \frac{\D B}{\D u^\ell} \D_x^{-1} u^j_x -\frac{1}{2} B^{-1} u^i_x \D_x^{-1} \frac{\D B}{\D u^k} g^{k\ell} \frac{\D B}{\D u^\ell}  u^j_x \\ \notag 
                                                                                                  & =
                                                                                                    -\frac{1}{2}  u^i_y \frac{\D B}{\D u^k} g^{k\ell} \frac{\D B}{\D u^\ell} \D_y^{-1} u^j_y -\frac{1}{2} u^i_y \D_y^{-1} \frac{\D B}{\D u^k} g^{k\ell} \frac{\D B}{\D u^\ell}  u^j_y \ ;
    \\
    B^{-1} B\delta^i_k P^{k\ell} \frac{\D B}{\D u^\ell} \D_x^{-1} u^j_x 
                                                                                                  & = P^{i\ell} \left(\frac{\D B}{\D u^\ell}\right) \D_x^{-1} u^j_x + g^{i\ell }\frac{\D B}{\D u^\ell} u^j_x \\ \notag 
                                                                                                  & 
                                                                                                    =  P^{i\ell} \left(\frac{\D B}{\D u^\ell}\right) \D_y^{-1} u^j_y + g^{i\ell }\frac{\D B}{\D u^\ell} B u^j_y \ ;
    \\
    -B^{-1} u^i_x \D_x^{-1} \frac{\D B}{\D u^k} P^{k\ell} B\delta^j_\ell 
                                                                                                  & = -B^{-1} u^i_x \D_x^{-1} (P^*)^{k j}\left(\frac{\D B}{\D u^k}\right) B -B^{-1} u^i_x \frac{\D B}{\D u^k} g^{kj} B
    \\ \notag 
                                                                                                  & =
                                                                                                    u^i_y \D_y^{-1} P^{jk}\left(\frac{\D B}{\D u^k}\right)  - u^i_y \frac{\D B}{\D u^k} g^{kj} B \ ; 
    \\
    B^{-1} B\delta^i_k  P^{k\ell} B\delta^j_\ell
                                                                                                  &=  P^{ij} B \ .
  \end{align}
  Thus the non-local term is given by Equation~\eqref{eq:F-P-nonloc}, and the
  local term is given by
  \begin{align}
    & P^{ij} B + g^{i\ell }\frac{\D B}{\D u^\ell} B u^j_y - u^i_y \frac{\D B}{\D u^k} g^{kj} B \\ \notag &
                                                                                                           =
                                                                                                           g^{ij}B^2 \D_y + g^{ij} B \frac{\D B}{\D u^k} u^k_y + \Gamma^{ij}_k B^2 u^k_y - \frac 12 g^{i\ell } B^4 \frac{\D B^{-2}}{\D u^\ell}  u^j_y
                                                                                                           +\frac 12 u^i_y \frac{\D B^{-2}}{\D u^k} g^{kj} B^4,
  \end{align}
  where the latter expression is equal to~\eqref{eq:F-P-loc}.
\end{proof}

\subsection{Weakly non-local bi-vectors of localizable shape}
The goal of this Section is to prove that the space of weakly non-local
bi-vectors of localizable shape is closed under the action of reciprocal
differential substitutions. We narrow the scope to the Miura-type substitutions
$\cR$ as in Definition~\ref{def:Miura-reciprocal-trans}.

Let us consider the effect of a reciprocal transformation of the form
\eqref{eq:change-x} on a general weakly nonlocal bi-vector of localizable
shape:
\begin{equation} \label{eq:bi-vector-loc-shape} P^{ij}=\sum_{d=1}^\infty
  \epsilon^{d-1} \left(\sum_{s=0}^d P^{ij}_{d,d-s} \D_x^s + u^{i,1} \D_x^{-1}
    V_d^j + V_d^i \D_x^{-1} u^{j,1}\right)=P_{loc}^{ij}+P_{nonloc}^{ij},
\end{equation}
where $P^{ij}_{d,d-s}\in\cA_{d-s}$ and $V^i_d \in \cA_d$. Note that both
$P_{loc}^{ij}$ and $P_{nonloc}^{ij}$ define skew-symmetric bi-vectors, that is
$(P_{loc}^*)^{ij} = - P_{loc}^{ji}$ and
$(P_{nonloc}^*)^{ij} = - P_{nonloc}^{ji}$.

\begin{proposition} Consider a Miura-reciprocal transformation in $\cR$ given
  by $dy = Bdx$, $w^i = Q^i$. Under this transformation any weakly non-local
  bi-vector $P^{ij}$ of localizable shape~\eqref{eq:bi-vector-loc-shape} is
  transformed into a weakly non-local bi-vector of localizable shape.
\end{proposition}

\begin{remark} In principle, this proposition follows from the arguments
  of~\cite{LiuZhang:JacobiStructures}
  and~\cite{ferapontov03:_recip_hamil}. However, it can also be directly
  obtained using Theorem~\ref{thm:Miura-Reciprocal-Bivector}.
\end{remark}

\begin{proof} 
  Repeating \emph{mutatis mutandis} the proof of Proposition~\ref{prop:F-P} one
  can check that the local part $P_{loc}^{ij}$ produces a weakly non-local
  operator of localizable shape (the only thing that matters for that
  computation is skew-symmetry of the bi-vector defined by $P_{loc}^{ij}$). So
  let us focus on the non-local part
  $P_{nonloc}^{ij} = u^{i}_x \D_x^{-1} V^j + V^i \D_x^{-1} u^{j}_x$, where
  $V^i = \sum_{d=1}^\infty \epsilon^{d-1} V^i_d$ and we use
  Convention~\ref{conv:ux} here and below in computations.  We have:
  \begin{align} \label{eq:PnonlocTrans} & B^{-1} \left(B
      \pd{Q^i}{u^{k,s}}\D_x^s - Q^i_x \D_x^{-1} \frac{\D B}{\D u^{k,s}}\D_x^s
    \right)\circ \left(u^{k}_x \D_x^{-1} V^l + V^k \D_x^{-1} u^{l}_x\right) \\
    \notag & \qquad \qquad \left((-\D_x)^t\circ \pd{Q^j}{u^{l,t}} B +
      (-\D_x)^t\circ \pd{B}{u^{l,t}} \D_x^{-1} Q^j_x \right)
  \end{align}
  (we omit the summation over $s$ and $t$ for brevity).  We
  compute~\eqref{eq:PnonlocTrans} as follows. First, note that
  \begin{align} \label{eq:WNLLoc-First} & \pd{Q^i}{u^{k,s}}\D_x^s \circ u^{k}_x
    \D_x^{-1} V^l (-\D_x)^t\circ \pd{Q^j}{u^{l,t}} B = Q^i_x \D_x^{-1}
    (\ell_{Q^j})_l(V^l) B + \mathsf{loc} ; \\ \notag & \pd{Q^i}{u^{k,s}}\D_x^s
    \circ V^{k} \D_x^{-1} u^l_x (-\D_x)^t\circ \pd{Q^j}{u^{l,t}} B =
    (\ell_{Q^i})_k(V^k) \D_x^{-1} Q^j_x B + \mathsf{loc}.
  \end{align}
  Here $ \mathsf{loc}$ are the terms where we collect some purely local
  operators. Furthermore,
  \begin{align} \label{eq:WNLLoc-Second} - \frac 1B Q^i_x \D_x^{-1} \frac{\D
      B}{\D u^{k,s}}\D_x^s \circ u^{k}_x \D_x^{-1} V^l (-\D_x)^t\circ
    \pd{Q^j}{u^{l,t}} B & = - \frac 1B Q^i_x \D_x^{-1} B_x \D_x^{-1}
    (\ell_{Q^j})_l(V^l) B - \frac 1B Q^i_x \D_x^{-1} O^j_{BuVQ} B; \\ \notag -
    \frac 1B Q^i_x \D_x^{-1} \frac{\D B}{\D u^{k,s}}\D_x^s \circ V^k \D_x^{-1}
    u^{l}_x (-\D_x)^t\circ \pd{Q^j}{u^{l,t}} B & = - \frac 1B Q^i_x \D_x^{-1}
    (\ell_B)_k(V^k) \D_x^{-1} {Q^j}_x B - \frac 1B Q^i_x \D_x^{-1} O^j_{BVuQ}
    B; \\ \notag \pd{Q^i}{u^{k,s}}\D_x^s \circ u^{k}_x \D_x^{-1} V^l
    (-\D_x)^t\circ \pd{B}{u^{l,t}} \D_x^{-1} Q^j_x & = Q^i_x \D_x^{-1}
    (\ell_B)_l(V^l) \D_x^{-1} Q^j_x + O^i_{QuVB} \D_x^{-1} Q^j_x; \\ \notag
    \pd{Q^i}{u^{k,s}}\D_x^s \circ V^k \D_x^{-1} u^{l}_x (-\D_x)^t\circ
    \pd{B}{u^{l,t}} \D_x^{-1} Q^j_x &= (\ell_{Q^i})_l(V^l) \D_x^{-1}
    B_x\D_x^{-1} Q^j_x + O^i_{QVuB}\D_x^{-1} Q^j_x.
  \end{align}
  Here $O^j_{BuVQ}$, $O^j_{BVuQ}$, $O^i_{QuVB}$, and $O^i_{QVuB}$ are some
  scalar local operators, whose main property is that
  $(O^j_{BuVQ})^* = - O^j_{QVuB}$ and $(O^j_{BVuQ})^* = - O^j_{QuVB}$. We omit
  their explicit formulas. Finally,
  \begin{align} \label{eq:WNLLoc-Third} - \frac 1B Q^i_x \D_x^{-1} \frac{\D
      B}{\D u^{k,s}}\D_x^s \circ u^{k}_x \D_x^{-1} V^l (-\D_x)^t\circ
    \pd{B}{u^{l,t}} \D_x^{-1} Q^j_x & = - \frac 1B Q^i_x \D_x^{-1} B_x
    \D_x^{-1} (\ell_{B})_l(V^l) \D_x^{-1} Q^j_x \\ \notag & \quad - \frac 1B
    Q^i_x \D_x^{-1} O_{BuVB} \D_x^{-1} Q^j_x; \\ \notag - \frac 1B Q^i_x
    \D_x^{-1} \frac{\D B}{\D u^{k,s}}\D_x^s \circ V^k \D_x^{-1} u^{l}_x
    (-\D_x)^t\circ \pd{B}{u^{l,t}} \D_x^{-1} Q^j_x & =- \frac 1B Q^i_x
    \D_x^{-1} (\ell_{B})_l(V^l) \D_x^{-1} B_x \D_x^{-1} Q^j_x \\ \notag & \quad
    - \frac 1B Q^i_x \D_x^{-1} O_{BVuB} \D_x^{-1} Q^j_x,
  \end{align}
  where $O_{BuVB}$ and $O_{BVuB}$ are scalar local operators such that
  $O_{BuVB}^*=-O_{BVuB}$. We omit their explicit formulas, but we use below
  that $O_{BuVB}+ O_{BVuB} = -\tilde O^* \D_x - \D_x\circ \tilde O$ for some
  local operator $\tilde O$.

  Now we collect the terms together. Firstly, we list all terms with $B_x$ that
  emerged in~\eqref{eq:WNLLoc-Second} and~\eqref{eq:WNLLoc-Third}:
  \begin{align} \label{eq:WNLLoc-Forth} - \frac 1B Q^i_x \D_x^{-1} B_x
    \D_x^{-1} (\ell_{Q^j})_l(V^l) B & = - Q^i_x \D_x^{-1} (\ell_{Q^j})_l(V^l) B
    + \frac 1B Q^i_x \D_x^{-1} (\ell_{Q^j})_l(V^l) B^2 \\ \notag
    (\ell_{Q^i})_l(V^l) \D_x^{-1} B_x\D_x^{-1} Q^j_x & = (\ell_{Q^i})_l(V^l)
    B\D_x^{-1} Q^j_x - (\ell_{Q^i})_l(V^l) \D_x^{-1} Q^j_x B \\ \notag - \frac
    1B Q^i_x \D_x^{-1} B_x \D_x^{-1} (\ell_{B})_l(V^l) \D_x^{-1} Q^j_x & = -
    Q^i_x \D_x^{-1} (\ell_{B})_l(V^l) \D_x^{-1} Q^j_x + \frac 1B Q^i_x
    \D_x^{-1} B (\ell_{B})_l(V^l) \D_x^{-1} Q^j_x \\ \notag - \frac 1B Q^i_x
    \D_x^{-1} (\ell_{B})_l(V^l) \D_x^{-1} B_x \D_x^{-1} Q^j_x & = - \frac 1B
    Q^i_x \D_x^{-1} (\ell_{B})_l(V^l) B \D_x^{-1} Q^j_x + \frac 1B Q^i_x
    \D_x^{-1} (\ell_{B})_l(V^l) \D_x^{-1} Q^j_x B
  \end{align}
  Note some cancellations: the non-local terms in~\eqref{eq:WNLLoc-First}
  cancel with the corresponding summands in the first and the second line
  of~\eqref{eq:WNLLoc-Forth}, two non-local terms in the second and third line
  of~\eqref{eq:WNLLoc-Second} cancel with the two terms in the third and forth
  line of~\eqref{eq:WNLLoc-Forth}, are there are two terms in the latter lines
  that cancel each other. So, modulo the purely local
  terms,~\eqref{eq:PnonlocTrans} is equal to the sum of the following four
  expressions:
  \begin{align}
    \frac 1B Q^i_x \D_x^{-1} (\ell_{Q^j})_l(V^l) B^2 + 	(\ell_{Q^i})_l(V^l) B\D_x^{-1} Q^j_x & = w^i_y \D_y^{-1} (\ell_{Q^j})_l(V^l) B + 	(\ell_{Q^i})_l(V^l) B\D_y^{-1} w^j_y;
    \\ \notag 
    \frac 1B Q^i_x \D_x^{-1} (O^j_{QVuB})^* B + O^i_{QVuB}\D_x^{-1} Q^j_x
                                                                                             & = \frac 1B Q^i_x \D_x^{-1} O^j_{QVuB}(1) B + O^i_{QVuB}(1) \D_x^{-1} Q^j_x + \mathsf{loc}
    \\ \notag & 
                = w^i_y \D_y^{-1} O^j_{QVuB}(1)  + O^i_{QVuB}(1) \D_y^{-1} w^j_y + \mathsf{loc};
    \\ \notag 
    \frac 1B Q^i_x \D_x^{-1} (O^j_{QuVB})^* B + O^i_{QuVB} \D_x^{-1} Q^j_x & = \frac 1B Q^i_x \D_x^{-1} O^j_{QuVB}(1) B + O^i_{QuVB}(1) \D_x^{-1} Q^j_x+ \mathsf{loc}
    \\ \notag &  = w^i_y \D_y^{-1} O^j_{QuVB}(1)  + O^i_{QuVB}(1) \D_y^{-1} w^j_y + \mathsf{loc};
    \\ \notag 
    - \frac 1B Q^i_x \D_x^{-1} O_{BuVB} \D_x^{-1} Q^j_x	- \frac 1B Q^i_x \D_x^{-1} O_{BVuB} \D_x^{-1} Q^j_x
                                                                                             & = \frac 1B Q^i_x \D_x^{-1}(\tilde O^* \D_x + \D_x\circ \tilde O) \D_x^{-1} Q^j_x	 
    \\ \notag
                                                                                             & =  \frac 1B Q^i_x   \D_x^{-1} \tilde O(1)Q^j_x	
                                                                                               + \frac 1B Q^i_x  \tilde O(1) \D_x^{-1} Q^j_x + \mathsf{loc}
    \\ \notag
                                                                                             & =  w^i_y   \D_y^{-1} 
                                                                                               \frac 1B\tilde O(1) Q^j_x	
                                                                                               + \frac 1B Q^i_x  \tilde O(1) \D_y^{-1} w^j_y + \mathsf{loc},	 	 
  \end{align}
  which is manifestly a weakly non-local operator of localizable shape.
\end{proof}

\section{Schouten bracket for weakly non-local operators of localizable shape}

\label{sec:LV-LZ}

The goal of this Section is to compare two ways to encode weakly non-local
Poisson structures of localizable shape: the one given in
\cite{LiuZhang:JacobiStructures} (by design only working for the localizable
shape case) and \cite{LORENZONI2020103573} (it is working for general weakly
non-local case, but we specialize it for the localizable shape). In principle,
the identification of these two approaches follows from the uniqueness property
of the bracket, c.f.~\cite[Theorem 2.4.1]{LiuZhang:JacobiStructures}, but we
want to present an explicit computation for this identification.

\subsection{The two approaches}
In both approaches the weakly non-local $p$-vectors of localizable shape are
encoded as
\begin{align}
  \label{eq:26}
  \int P= \int P_L+\zeta P_N,
\end{align}
where $ P_L\in \hat{\cA}^p$, $P_N\in \hat{\cA}^{p-1}$, and
\begin{equation}
  \label{eq:MainPropertyofZeta}
  \partial_x\zeta = -u^{i,1}\theta_i.
\end{equation}
The difference in two approaches is the meaning of $\zeta$. In the approach
of~\cite{LiuZhang:JacobiStructures}, $\zeta$ is a new dependent variable such
that $\deg_{\D_x} \zeta = 0$ and $\deg_{\theta}\zeta =1$. The new space of
multivector densities is defined as $\cS\coloneqq \hat{\cA}[\zeta]$, equipped
with the operator
\begin{align}
  \D_x & = -u^{i,1}\theta_i \D_\zeta + \sum u^{i,d+1} \D_{u^{i,d}} + \theta_i^{d+1}\D_{\theta_i^{d}},
\end{align}
and the space of weakly non-local multivectors of localizable shape is defined
as $\cE\coloneqq \cS/\D_x\cS$.

In the approach of~\cite{LORENZONI2020103573}, $\zeta$ is not a new dependent
variable, but rather an expression in the existing dependent variables (still
of differential degree $\deg_{\D_x} \zeta = 0$ and multivector degree
$\deg_{\theta}\zeta =1$), such that Equation~\eqref{eq:MainPropertyofZeta} is
satisfied for the standard operator
\begin{align}
  \tilde{\D}_x & =  \sum u^{i,d+1} \D_{u^{i,d}} + \theta_i^{d+1}\D_{\theta_i^{d}}.
\end{align}
For instance, one can find such a function in $\hat{\cA}((\frac 1{u^{1,1}}))$,
cf.~\cite{DubrovinLiuZhang:QuasiTriv}. To this end, one looks for a unique
solution $\tilde{\D}_x \zeta = - u^{i,1}\theta_i$ of the form
$\zeta = \sum_{i=1}^\infty \tfrac{f_i} {(u^{1,1})^{i}}$, with
$f_i\in \hat{\cA}$ such that $\D_{u^{1,1}} f_i = 0$.

Once the objects are defined, we have two different formulae for the Schouten
bracket in these two approaches:
\begin{itemize}
\item The formula in the approach of~\cite{LORENZONI2020103573} is
  \begin{equation}
    \Big[\int P,\int Q\Big]=\int (-1)^{\deg_\theta P}\tilde\delta_{u^i}P\tilde\delta_{\theta_i}{Q} +
    \tilde\delta_{\theta_i}P\tilde\delta_{u^i}{Q}.
  \end{equation}
  Recall that $\zeta$ is regarded as a function of
  $(u^i_\sigma,\theta_i^{\sigma})$ in the variational derivatives (which are
  denoted by $\tilde\delta_{u^i}$ and $\tilde\delta_{\theta_i}$ for that
  reason).
\item The formula in the approach of~\cite{LiuZhang:JacobiStructures} is
  \begin{equation}
    \Big[\int P,\int Q\Big]=\int (-1)^{\deg_\theta P}\delta_{u^i}{P}\delta_{\theta_i}{Q} +
    \delta_{\theta_i}{P}\delta_{u^i}{Q}
    + (-1)^{\deg_\theta P}\hat{E}(P)\partial_{\zeta}Q +
    \partial_{\zeta}P\hat{E}(Q)
  \end{equation}
  Here $\zeta$ is regarded as an extra dependent variable, and the operator
  $\hat{E}$ is defined as
  \begin{align}
    \hat E & = \sum_{\substack{s\geq 1 \\ t\geq 0}} \Big(
    u^{i,s} (-\D_x)^t \D_{u^{i,s+t}} + \theta_i^{s} (-\D_x)^t \D_{\theta_i^{s+t}} 
    \Big) - 1 + \theta_i \delta_{\theta_i}.
  \end{align}
\end{itemize}

\subsection{Identification of the two approaches} We prove the following:

\begin{theorem} The identity map $\hat{\cA}[\zeta]\to \hat{\cA}[\zeta]$ induces
  the isomorphism of the Lie algebras of local multivector fields defined by the
  Schouten brackets in these two approaches.
\end{theorem}
 
\begin{proof}
  We represent any density $P\in \hat{\cA}[\zeta]$ as $P=P_L+\zeta P_N$ and
  consider $\zeta$ to be a nonlocal function. Note that
  \begin{align}
    \label{eq:34}
    \tfdS{P}{u^i} & = \fdS{P}{u^i}
                    + ( - \partial_x)^\sigma\left(\pdS{\zeta}{u^{i,\sigma}}P_N\right)
    \\ \notag 
                  & = \fdS{P_L}{u^i} +
                    ( - \partial_x)^\sigma\left(\zeta\pdS{P_N}{u^{i,\sigma}}\right) +
                    ( - \partial_x)^\sigma\left(\pdS{\zeta}{u^{i,\sigma}}P_N\right),
    \\
    \tfdS{P}{\theta_i} & =  \fdS{P}{\theta_i}
                         + ( - \partial_x)^\sigma\left(\pdS{\zeta}{\theta_i^\sigma}P_N\right)
    \\ \notag 
                  & =\fdS{P_L}{\theta_i} -
                    ( - \partial_x)^\sigma\left(\zeta\pdS{P_N}{\theta_i^\sigma}\right) +
                    ( - \partial_x)^\sigma\left(\pdS{\zeta}{\theta_i^\sigma}P_N\right),
  \end{align}
  where we used that
  \begin{align}
    \label{eq:34b}
    \fdS{P}{u^i} =& \fdS{P_L}{u^i} +
                    ( - \partial_x)^\sigma
                    \left(\zeta\pdS{P_N}{u^{i,\sigma}}\right),
    \\
    \fdS{P}{\theta_i} =& \fdS{P_L}{\theta_i} -
                         ( - \partial_x)^\sigma
                         \left(\zeta\pdS{P_N}{\theta_i^\sigma}\right).
  \end{align}
%
  Using these formulas, we obtain
  \begin{align}
    \label{eq:31}
    \tfdS{P}{u^i}\tfdS{Q}{\theta_i} & =
                                      \left(\fdS{P}{u^i}
                                      + ( - \partial_x)^\sigma\left(\pdS{\zeta}{u^{i,\sigma}}P_N\right)\right)
                                      \left( \fdS{Q}{\theta_i}
                                      + ( - \partial_x)^\sigma\left(\pdS{\zeta}{\theta_i^\sigma}Q_N\right)\right)
    \\ \notag 
                                    & =\fdS{P}{u^i} \fdS{Q}{\theta_i} + \fdS{P}{u^i}( -
                                      \partial_x)^\sigma\left(\pdS{\zeta}{\theta_i^\sigma}Q_N\right)
                                      + ( - \partial_x)^\sigma\left(\pdS{\zeta}{u^{i,\sigma}}P_N\right)
                                      \fdS{Q}{\theta_i}
    \\ \notag 
                                    &\quad 
                                      + ( - \partial_x)^\sigma\left(\pdS{\zeta}{u^{i,\sigma}}P_N\right)
                                      ( - \partial_x)^\sigma\left(\pdS{\zeta}{\theta_i^\sigma}Q_N\right);
    \\
    \tfdS{P}{\theta_i}\tfdS{Q}{u^i} & =
                                      \left(\fdS{P}{\theta_i}
                                      + ( - \partial_x)^\sigma\left(\pdS{\zeta}{\theta_i^\sigma}P_N\right)\right)
                                      \left( \fdS{Q}{u^i}
                                      + ( - \partial_x)^\sigma\left(\pdS{\zeta}{u^{i,\sigma}}Q_N\right)\right)
    \\ \notag 
                                    & =\fdS{P}{\theta_i} \fdS{Q}{u^i} + \fdS{P}{\theta_i}( -
                                      \partial_x)^\sigma\left(\pdS{\zeta}{u^{i,\sigma}}Q_N\right)
                                      + ( - \partial_x)^\sigma\left(\pdS{\zeta}{\theta_i^\sigma}P_N\right)
                                      \fdS{Q}{u^i}
    \\ \notag 
                                    & \quad 
                                      + ( - \partial_x)^\sigma\left(\pdS{\zeta}{\theta_i^\sigma}P_N\right)
                                      ( - \partial_x)^\sigma\left(\pdS{\zeta}{u^{i,\sigma}}Q_N\right).
  \end{align}

  If we want to treat $\zeta$ as a new dependent variable, we have
  $\hat{E}(P)\partial_{\zeta}Q = \hat{E}(P)Q_N$ (and similarly for the other
  summand in the formula), so we have to prove that
  \begin{align}
    \label{eq:ProofLVLZ-1}
    & \int (-1)^{\deg_{\theta}P}\hat{E}(P)Q_N + P_N\hat{E}(Q)
    \\ \notag &=
                \int (-1)^{\deg_\theta P}\Big( \fdS{P}{u^i}( -
                \partial_x)^\sigma\left(\pdS{\zeta}{\theta_i^\sigma}Q_N\right)
                + ( - \partial_x)^\sigma\left(\pdS{\zeta}{u^{i,\sigma}}P_N\right)
                \fdS{Q}{\theta_i}
    \\ \notag  & \qquad \qquad \qquad \quad 
                 + ( - \partial_x)^\sigma\left(\pdS{\zeta}{u^{i,\sigma}}P_N\right)
                 ( - \partial_x)^\sigma\left(\pdS{\zeta}{\theta_i^\sigma}Q_N\right) \Big)
    \\ \notag & \qquad
                +\Big( \fdS{P}{\theta_i}( -
                \partial_x)^\sigma\left(\pdS{\zeta}{u^{i,\sigma}}Q_N\right)
                + ( - \partial_x)^\sigma\left(\pdS{\zeta}{\theta_i^\sigma}P_N\right)
                \fdS{Q}{u^i}
    \\
    & \notag  \qquad \qquad
      + ( - \partial_x)^\sigma\left(\pdS{\zeta}{\theta_i^\sigma}P_N\right)
      ( - \partial_x)^\sigma\left(\pdS{\zeta}{u^{i,\sigma}}Q_N\right) \Big).
  \end{align}
  Let us use the following property of the operator $\hat{E}$:
  \begin{align}
    \D_x\hat E = -u^{i,1} \delta_{u^i} + \theta_i \D_x \delta_{\theta_i} + u^{i,1}\theta_i \delta_\zeta; 
  \end{align}
  So, we obtain
  \begin{align}
    \label{eq:hatEPQN}
    \int \hat{E}(P)Q_N & =
                         \int
                         \partial_x^{-1}\left(-u^{i,1}\fdS{P}{u^i}
                         +\theta_{i}\partial_x\fdS{P}{\theta_i}
                         + u^{i,1}\theta_iP_N\right)Q_N
    \\ \notag 
                       & = \int - \left(-u^{i,1}\fdS{P}{u^i}
                         +\theta_{i}\partial_x\fdS{P}{\theta_i}
                         + u^{i,1}\theta_iP_N\right) \partial_x^{-1}(Q_N),
    \\ \label{eq:PNhatEQ}
    \int P_N\hat{E}(Q) & =
                         \int 
                         P_N\partial_x^{-1}\left(-u^{i,1}\fdS{Q}{u^i}
                         +\theta_{i}\partial_x\fdS{Q}{\theta_i}
                         +  u^{i,1}\theta_iQ_N\right)
    \\ \notag 
    =& \int - \partial_x^{-1}(P_N)\left(-u^{i,1}\fdS{Q}{u^i}
       +\theta_{i}\partial_x\fdS{Q}{\theta_i}
       +  u^{i,1}\theta_iQ_N\right).
  \end{align}
  Substituting Equations~\eqref{eq:hatEPQN} and~\eqref{eq:PNhatEQ}
  into~\eqref{eq:ProofLVLZ-1}, we see that the statement of the theorem reduces
  to the following equality:
  \begin{align}
    \label{eq:ProofLVLZ-2}
    & \int (-1)^{\deg_{\theta}P+1} \left(-u^{i,1}\fdS{P}{u^i}
      +\theta_{i}\partial_x\fdS{P}{\theta_i}
      + u^{i,1}\theta_iP_N\right) \partial_x^{-1}(Q_N)
    \\ \notag 
    &  \quad       - \partial_x^{-1}(P_N)\left(-u^{i,1}\fdS{Q}{u^i}
      +\theta_{i}\partial_x\fdS{Q}{\theta_i}
      +  u^{i,1}\theta_iQ_N\right)
    \\ \notag &=
                \int (-1)^{\deg_\theta P}\Big( \fdS{P}{u^i}( -
                \partial_x)^\sigma\left(\pdS{\zeta}{\theta_i^\sigma}Q_N\right)
                + ( - \partial_x)^\sigma\left(\pdS{\zeta}{u^{i,\sigma}}P_N\right)
                \fdS{Q}{\theta_i}
    \\ \notag  & \qquad \qquad \qquad \quad 
                 +( - \partial_x)^\sigma\left(\pdS{\zeta}{u^{i,\sigma}}P_N\right)
                 ( - \partial_x)^\sigma\left(\pdS{\zeta}{\theta_i^\sigma}Q_N\right) \Big)
    \\ \notag & \qquad
                +\Big( \fdS{P}{\theta_i}( -
                \partial_x)^\sigma\left(\pdS{\zeta}{u^{i,\sigma}}Q_N\right)
                + ( - \partial_x)^\sigma\left(\pdS{\zeta}{\theta_i^\sigma}P_N\right)
                \fdS{Q}{u^i}
    \\
    & \notag  \qquad \qquad 
      + ( - \partial_x)^\sigma\left(\pdS{\zeta}{\theta_i^\sigma}P_N\right)
      ( - \partial_x)^\sigma\left(\pdS{\zeta}{u^{i,\sigma}}Q_N\right) \Big).
  \end{align}

  In order to prove this equality, our strategy is move $\partial_x^{-1}$ in
  $\zeta = \partial_x^{-1}(-u^{i,1}\theta_i)$ to the other factor ($P_N$ or
  $Q_N$) using integration by parts. We have:
  \begin{align}
    \label{eq:52}
    \fdS{P}{u^i}( -
    \partial_x)^\sigma\left(\pdS{\zeta}{\theta_i^\sigma}Q_N\right) & =
                                                                     \fdS{P}{u^i}u^{i,1}\partial_x^{-1}(Q_N),
    \\
    ( - \partial_x)^\sigma\left(\pdS{\zeta}{u^i_\sigma}P_N\right)
    \fdS{Q}{\theta_i} & =
                        - \partial_x\left(\theta_i\partial_x^{-1}(P_N)\right)
                        \fdS{Q}{\theta_i}
    \\
    ( - \partial_x)^\sigma\left(\pdS{\zeta}{u^i_\sigma}P_N\right)
    ( - \partial_x)^\sigma\left(\pdS{\zeta}{\theta_i^\sigma}Q_N\right) & =
                                                                         - \partial_x\left(\theta_i
                                                                         \partial_x^{-1}(P_N)\right)
                                                                         u^{i,1}\partial_x^{-1}(Q_N)
    \\
    ( - \partial_x)^\sigma\left(\pdS{\zeta}{\theta_i^\sigma}P_N\right)
    \fdS{Q}{u^i} & =
                   u^{i,1}\partial_x^{-1}(P_N)\
                   \fdS{Q}{u^i}
    \\
    \fdS{P}{\theta_i}( - \partial_x)^\sigma
    \left(\pdS{\zeta}{u^i_\sigma}Q_N\right) & =
                                              -\fdS{P}{\theta_i} \partial_x
                                              \left(\theta_i\partial_x^{-1}(Q_N)\right)
    \\
    ( - \partial_x)^\sigma\left(\pdS{\zeta}{\theta_i^\sigma}P_N\right)
    ( - \partial_x)^\sigma\left(\pdS{\zeta}{u^i_\sigma}Q_N\right) & =
                                                                    -\left(u^{i,1}\partial_x^{-1}(P_N)\right)
                                                                    \partial_x\left(\theta_i\partial_x^{-1}(Q_N)\right)
  \end{align}
  Substituting the above expressions into the equality~\eqref{eq:ProofLVLZ-2}
  that we shall prove, we are led to the simplified equality:
  \begin{align}
    & \int (-1)^{\deg_\theta P+1}\left(\theta_{i}\partial_x\fdS{P}{\theta_i}
      + u^{i,1}\theta_iP_N\right)\partial_x^{-1}(Q_N) 
      - \partial_x^{-1}(P_N)\left(
      \theta_{i}\partial_x\fdS{Q}{\theta_i}
      +  u^{i,1}\theta_iQ_N\right)
    \\ \notag 
    & = \int (-1)^{\deg_\theta P+1}\Big(\partial_x(\theta_i\partial_x^{-1}(P_N))
      \fdS{Q}{\theta_i}
      + \partial_x(\theta_i\partial_x^{-1}(P_N))
      u^{i,1}\partial_x^{-1}(Q_N)\Big)
    \\ \notag & \qquad 
                - \Big( \fdS{P}{\theta_i} \partial_x(\theta_i\partial_x^{-1}(Q_N))
                +u^{i,1}\partial_x^{-1}(P_N)\partial_x(\theta_i\partial_x^{-1}(Q_N))\Big).
  \end{align}
  Integrating by parts the summands containing $\partial_x\fdS{P}{\theta_i}$,
  $\partial_x\fdS{Q}{\theta_i}$ we obtain the further simplification of the
  equality~\eqref{eq:ProofLVLZ-2} (note that
  $\deg_\theta P_N=\deg_\theta P-1$):
  \begin{align}
    \label{eq:512}
    & \int (-1)^{\deg_\theta P+1}u^{i,1}\theta_iP_N\partial_x^{-1}(Q_N) - \partial_x^{-1}(P_N)u^{i,1}\theta_iQ_N
    \\ \notag 
    & = \int (-1)^{\deg_\theta P+1 } \partial_x(\theta_i\partial_x^{-1}(P_N))
      u^{i,1}\partial_x^{-1}(Q_N)
      - u^{i,1}\partial_x^{-1}(P_N)\partial_x(\theta_i\partial_x^{-1}(Q_N))
  \end{align}
  Expanding the total derivatives on the right-hand side we easily see that the
  above equality is an identity. This completes the proof of the theorem.
\end{proof}


\section{Pencils of weakly non-local bi-vectors of localizable shape}

\label{sec:pencils}

In this Section we compute the bi-Hamiltonian cohomology for a semi-simple
pencil of weakly non-local Poisson bi-vectors of localizable shape of
differential order $\deg_{\D_x}=1$ satisfying the extra condition: the pencil
of these bi-vectors should be localizable (or, equivalently, they should be
simultaneously localizable) with respect to the Miura-reciprocal group. As a
result of this computation and some further arguments we prove the following
theorem:

\begin{theorem} \label{thm:LocalLeadingTermOfAPencil}
	
  Let $P_1$ and $P_2$ be two of commuting non-local Poisson bi-vectors of
  localizable shape. We assume that $P_1$ and $P_2$ have dispersive expansion
  given by $P_a=\sum_{i=1}^\infty \epsilon^{i-1} P_{a,i}$,
  $\deg_{\D_x} P_{a,i} = i$, $a=1,2$, $i=1,2,\dots$.

  If the leading terms of degree $\deg_{\D_x}=1$, $P_{1,1}$ and $P_{2,1}$, are
  simultaneously localizable under the action of the Miura-reciprocal group and
  form a semi-simple Poisson pencil, then the full dispersive brackets $P_1$
  and $P_2$ are simultaneously localizable under the action of the
  Miura-reciprocal group.
\end{theorem}

In order to prove this theorem, we have to make a few preliminary computations
with bi-Hamiltonian cohomology, following the ideas
in~\cite{LiuZhang:JacobiStructures} subsequent steps
in~\cite{CarletPosthumaShadrin:DeformSS,CKS}.

\subsection{Bi-Hamiltonian cohomology}


\subsubsection{Setup for a deformation problem} 
Recall that following Liu and Zhang~\cite{LiuZhang:JacobiStructures} we denote
$\cS \coloneqq \hat\cA[\zeta]$, with $\D_x\colon\cS\to\cS$ given by
$\D_x = -u^{i,1}\theta_i \D_\zeta + \sum u^{i,d+1} \D_{u^{i,d}} +
\theta_i^{d+1} \D_{\theta_i^{d}}$, and $\cE\coloneqq \cS/\D_x\cS$.

Let $P_1, P_2\in S^2_1$ such that $\int P_1$ and $\int P_2$ form a pencil of
Poisson structures (possibly non-local, but then they are automatically weakly
non-local of localizable shape, since it is the only type of non-locality
accommodated in the space $\cE$), that is, we assume that
\begin{align}
  \Big[\int P_2-\lambda P_1, \int P_2-\lambda P_1\Big] = 0
\end{align}

Recall that there is a group $\mathcal{R}_I$ of the Miura-reciprocal
transformations of the 1st kind acting on them, see
Equation~\eqref{eq:Miura-reciprocal-1st-kind}. We assume that the pencil
$\int P_2-\lambda P_1$ is localizable under the action of $\mathcal{R}_I$. We
also assume that the pencil formed by $P_1$ and $P_2$ is semi-simple, which
together with the assumption of localizability implies that the we can choose
the coordinates $x,u^1,\dots,u^N$ such that the densities $P_1$ and $P_2$ of
the bivectors $\int P_1$ and $\int P_2$ take the form
\begin{align}
  P_1 & = \Big(\sum_{i=1}^N f^i \theta_i\theta_i^1\Big)
        + \Gamma^{ij}_{1,k} u^{k,1} \theta_i\theta_j; \\
  P_2 & =\Big( \sum_{i=1}^N u^i f^i \theta_i\theta_i^1\Big)
        + \Gamma^{ij}_{2,k} u^{k,1} \theta_i\theta_j. 
\end{align} 

We are interested to classify the equivalence classes of the higher order
dispersive deformations of the Poisson pencil $\int P_2-\lambda P_1$ in $\cE$
with respect to the Miura-reciprocal transformations of the 2nd kind,
$\cR_{II}$. Let $d_i\coloneqq \ad_{P_i} \colon \cE\to \cE$, $i=1,2$. Then the
deformation problem is controlled by the bi-Hamiltonian cohomology
$BH^{p}_d(\cE,d_1,d_2)$ of cohomological degree $p=2$ and $p=3$ and of
differential degrees $d\geq 2$ and $d\geq 4$, respectively. It is a rather
standard argument, see e.~g.~\cite[Proposition
3.3.5]{LiuZhang:JacobiStructures}. The only extra bit that one needs in our
case, that is, the space $\cE$ and the group $\cR_{II}$ of Miura-reciprocal
transformations of the 2nd kind, in comparison with the usual local case, that
is, the space $\hat\cF$ and the group $\cG_{II}$ of Miura transformations of
the 2nd kind, is the identification of the action of the Lie algebra of
$\cR_{II}$ on weakly non-local bi-vectors (or, more generally, multivectors)
of localizable shape with the adjoint action of $\cE^1$ on $\cE^2$ (resp.,
$\cE$). This is established in~\cite[Theorems 2.5.7 and
2.6.5]{LiuZhang:JacobiStructures}

\subsubsection{Bi-Hamiltonian cohomology computation} We prove the following

\begin{theorem} \label{thm:BiHomCoh} We have:
  \begin{align}
    BH^2_d(\cE,d_1,d_2) & \cong \begin{cases}
      0, & d=2 \text{ and } d\geq 4;\\
      \bigoplus_{i=1}^N C^\infty(\mathbb{R},u^i), & d=3. 
    \end{cases} \\
    BH^3_d(\cE,d_1,d_2) & \cong 0, \qquad d \geq 4.
  \end{align}
\end{theorem}

\begin{proof}
  For the proof we use that for $d\geq 2$ we have \cite[Lemma 4.4]{MR3123540}:
  \begin{align} \label{eq:iso-BH-bipol} BH^{p}_d(\cE,d_1,d_2)\cong
    H^p_d(\cE[\lambda], d_2-\lambda d_1).
  \end{align}
  In order to compute $H^p_d(\cE[\lambda], d_2-\lambda d_1)$, we recall the
  definition of $D_i\coloneqq D_{P_i}\colon\cS\to\cS$
  from~\cite{LiuZhang:JacobiStructures}:
  \begin{align}
    D_{P_i} &\coloneqq \hat E(P_i) \D_\zeta + \sum_{s=0}^\infty \D_x^s\big(\delta_{u^j} P_i \big)\D_{\theta_j^s}
              +\D_x^s\big(\delta_{ \theta_j} P_i \big)\D_{u^{j,s}}, && i=1,2.
  \end{align}
  Note that $[\D_x, D_i] = 0$ (by direct computation). We prove that it is a
  homological vector field (which is not true in general, for a non-local
  bi-vector $P_i$):

  \begin{lemma}\label{lem:PropertiesD_P} For a purely local bivector $\int P$
    the operator $D_{P}$ does not depend on the choice of a purely local
    density $P$. Moreover, for purely local densities of the bivectors
    $P,Q\in \hat \cA^2$ and for any $T\in \cS$ we have:
    \begin{align} \label{eq:lemma-DP-action} \int D_{P}(T) & = \Big[\int P,
      \int T\Big]
    \end{align}
    and
    \begin{align}
      \label{eq:lemma-DP-commutation}
      [D_{P},D_{Q}] &= D_{[P,Q]},
    \end{align}
    where
    $[P ,Q ] = \delta_{\theta_i} P \delta_{u^i} Q +\delta_{u^i} P
    \delta_{\theta_i} Q $.

    In particular, for a purely local density $P$ of a Poisson bivector
    $\int P$ we have $D_{P}^2=0$ on $\cS$.
  \end{lemma}
  \begin{remark}
    The statements of Lemma~\ref{lem:PropertiesD_P} do not hold for not purely
    local densities.
  \end{remark}

  \begin{proof}[Proof of Lemma~\ref{lem:PropertiesD_P}] Firstly, we check the
    $D_{P}$ does not depend on the choice of a local density $P$. To this end,
    we remind the definitions and basic properties of $\hat E$ and $\D_x$. We
    have:
    \begin{align}
      \D_x & = -u^{i,1}\theta_i \D_\zeta + \sum u^{i,d+1} \D_{u^{i,d}} + \theta_i^{d+1} \D_{\theta_i^{d}}; \label{eq:def-Partial_x-NL}
      \\
      \hat E & = \sum_{\substack{s\geq 1 \\ t\geq 0}} \Big(
      u^{i,s} (-\D_x)^t \D_{u^{i,s+t}} + \theta_i^{s} (-\D_x)^t \D_{\theta_i^{s+t}} 
      \Big) - 1 + \theta_i \delta_{\theta_i};
      \\
      \D_x\hat E & = -u^{i,1} \delta_{u^i} + \theta_i \D_x \delta_{\theta_i} + u^{i,1}\theta_i \delta_\zeta; 
      \\
      \hat E \D_x & = -u^{i,1}\theta_i \D_\zeta;
      \\
      \delta_{u^i} \D_x & = \D_x \theta_i\D_\zeta, &&  i=1,\dots,N; 
      \\
      \delta_{\theta_i} \D_x & = -u^{i,1}\D_\zeta, && i=1,\dots,N.
    \end{align}
    With the last three equations we immediately see that for any local
    $X\in \hat\cA$
    \begin{align}
      & \hat E(\D_x X) \D_\zeta + \sum_{s=0}^\infty \Big(\D_x^s\big(\delta_{u^j} \D_x X \big)\D_{\theta_j^s}
	+\D_x^s\big(\delta_{\theta_j} \D_x X \big)\D_{u^{j,s}}\Big) =
      \\ \notag 
      & -u^{i,1}\theta_i \D_\zeta X \D_\zeta + \sum_{s=0}^\infty\Big( \D_x^s\big(\D_x (\theta_i\D_\zeta X )\big)\D_{\theta_j^s}
	+\D_x^s\big(-u^{i,1}\D_\zeta X \big)\D_{u^{j,s}}\Big)
      = 0,
    \end{align}
    since $ \D_\zeta=0$, which implies the first assertion of the lemma.
	
    Now, Equation~\eqref{eq:lemma-DP-action} is obvious from the definition of
    the Schouten bracket. So we focus on
    Equation~\eqref{eq:lemma-DP-commutation}. Let us compute the coefficient of
    $\partial_\zeta$ on the left hand side. Using the vanishing of $\D_\zeta$
    derivatives, we have:
    \begin{align}
      & \sum_{s=0}^\infty \Big(\D_x^s\big(\delta_{u^j} P  \big)\D_{\theta_j^s} 
	+\D_x^s\big(\delta_{\theta_j} P  \big)\D_{u^{j,s}}\Big) 
        \hat E( Q ) =
      \\ \notag
      & \D_x^{-1} \sum_{s=0}^\infty \Big(\D_x^s\big(\delta_{u^j} P  \big)\D_{\theta_j^s} 
        +\D_x^s\big(\delta_{\theta_j} P  \big)\D_{u^{j,s}}\Big) 
        (-u^{i,1} \delta_{u^i} + \theta_i \D_x \delta_{\theta_i})( Q ) =
      \\ \notag &
                  \D_x^{-1}  \Big(\delta_{u^j} P   \D_x (\delta_{\theta_j} Q )
                  - \D_x(\delta_{\theta_j} P ) \delta_{u^j} Q  \Big) 
      \\ \notag & 
                  + \D_x^{-1} (-u^{i,1})  \sum_{s=0}^\infty \Big(\D_x^s\big(\delta_{u^j} P  \big)\D_{\theta_j^s} 
                  +\D_x^s\big(\delta_{\theta_j} P  \big)\D_{u^{j,s}}\Big) \delta_{u^i} Q 
      \\ \notag & 
                  + \D_x^{-1} (-\theta_i\D_x)  \sum_{s=0}^\infty \Big(\D_x^s\big(\delta_{u^j} P  \big)\D_{\theta_j^s} 
                  +\D_x^s\big(\delta_{\theta_j} P  \big)\D_{u^{j,s}}\Big) \delta_{\theta_i} Q 
    \end{align}
    Adding to the latter expression the same one with interchanged $P $ and
    $Q $ and using that for purely local densities
    \begin{align}
      & \sum_{s=0}^\infty \Big(\D_x^s\big(\delta_{u^j} P  \big)\D_{\theta_j^s} 
	+\D_x^s\big(\delta_{\theta_j} P  \big)\D_{u^{j,s}}\Big) \delta_{u^i} Q 
	+ \sum_{s=0}^\infty \Big(\D_x^s\big(\delta_{u^j} Q  \big)\D_{\theta_j^s} 
	+\D_x^s\big(\delta_{\theta_j} Q  \big)\D_{u^{j,s}}\Big) \delta_{u^i} P 
      \\ \notag 
      & = \delta_{u^i} \sum_{s=0}^\infty \Big(\D_x^s\big(\delta_{u^j} P  \big)\D_{\theta_j^s} Q 
	+\D_x^s\big(\delta_{\theta_j} P  \big)\D_{u^{j,s}}Q \Big) 
        = \delta_{u^i} \sum_{s=0}^\infty \Big(\delta_{u^j} P  \delta_{\theta_j} Q 
	+\delta_{\theta_j} P  \delta_{u^{j}}Q \Big) 
    \end{align}
    and
    \begin{align}
      & \sum_{s=0}^\infty \Big(\D_x^s\big(\delta_{u^j} P  \big)
        \D_{\theta_j^s} 
        +\D_x^s\big(\delta_{\theta_j} P  \big)\D_{u^{j,s}}\Big)
        \delta_{\theta_i} Q 
        + \sum_{s=0}^\infty \Big(\D_x^s\big(\delta_{u^j} Q  \big)\D_{\theta_j^s} 
        +\D_x^s\big(\delta_{\theta_j} Q  \big)\D_{u^{j,s}}\Big)
        \delta_{\theta_i} P 
      \\ \notag 
      & = -\delta_{\theta^i} \sum_{s=0}^\infty
        \Big(\D_x^s\big(\delta_{u^j} P  \big)\D_{\theta_j^s} Q 
	+\D_x^s\big(\delta_{\theta_j} P  \big)\D_{u^{j,s}}Q \Big) 
        = -\delta_{\theta_i} \sum_{s=0}^\infty
        \Big(\delta_{u^j} P  \delta_{\theta_j} Q 
	+\delta_{\theta_j} P  \delta_{u^{j}}Q \Big) ,
    \end{align}
    we obtain that the coefficient of $\D_\zeta$ on the left hand side of
    Equation~\eqref{eq:lemma-DP-commutation} is equal to
    \begin{align}
      \D_x^{-1} (-u^{i,1} \delta_{u^i}+\theta_i\D_x \delta_{\theta_i})
      [P ,Q ] = \hat E \big([P ,Q ]\big),
    \end{align}
    which is the coefficient of $\D_\zeta$ on the right hand side of
    Equation~\eqref{eq:lemma-DP-commutation}. The coefficients of all other
    components of the vector fields on the left hand side of
    Equation~\eqref{eq:lemma-DP-commutation} are computed in a very similar
    way.
  \end{proof}

  Lemma~\ref{lem:PropertiesD_P} implies that $D_2-\lambda D_1$ is a
  differential on $\cS[\lambda]$, and we have a short exact sequence
  \begin{align}
    \xymatrix{
    0 \ar[r] & \dfrac{\cS[\lambda]}{\mathbb{R}} 
               \ar[r]^{\D_x} 
               \ar@(u,ru)^{D_2-\lambda D_1}
    & \cS[\lambda]\ar[r]^{\int}
      \ar@(u,ru)^{D_2-\lambda D_1}
    & \cE[\lambda] \ar[r] 
      \ar@(u,ru)^{d _2-\lambda d _1}& 0 
                                      }
  \end{align}
  and it implies a long exact sequence in the cohomology which reads
  \begin{align} \label{eq:long-exact} \xymatrix{
      H^{p}_{d-1}(\cS[\lambda]/\mathbb{R}, D_2-\lambda D_1) \ar[r] &
      H^{p}_d(\cS[\lambda], D_2-\lambda D_1) \ar[r] & H^{p}_d(\cE[\lambda], d
      _2-\lambda d _1)
                                                      \ar@(d,u)[dll] \\
    H^{p+1}_{d}(\cS[\lambda]/\mathbb{R}, D_2-\lambda D_1) \ar[r] &
                                                                   H^{p+1}_{d+1}(\cS[\lambda],
                                                                   D_2-\lambda
                                                                   D_1) & }
  \end{align}

  \begin{lemma} \label{lem:CohS} We have
    \begin{align}
      H^p_d(\cS[\lambda], D_2-\lambda D_1) \cong
      \begin{cases}
        0 & p\leq d \text{ and } (p,d)\not= (3,3),(0,0) \\
        \mathbb{R}[\lambda] & p=0,d=0;\\
        \bigoplus_{i=1}^N C^\infty(\mathbb{R}, u^i) & p=3,d=3.
      \end{cases}
    \end{align}
    Also, $H^3_2(\cS[\lambda], D_2-\lambda D_1) \cong 0$.
  \end{lemma}
  \begin{proof} This lemma can be derived from~\cite[Theorems 2.12 and
    2.13]{CKS}. Indeed, Lemma~\ref{lem:PropertiesD_P} in particular implies
    that we have a bicomplex $(\cS[\lambda],D^{loc},D^{\zeta})$ with the
    differentials given by
    $D^\zeta\coloneqq (\hat{E}(P_2) - \lambda \hat E(P_1))\partial_\zeta$ and
    $D^{loc}\coloneqq D_2-\lambda D_1 - D^\zeta$. We start a spectral sequence
    associated with this bicomplex. Obviously, it converges on the second
    page. The computation of the first page splits as
    \begin{align}
      H^p_d(\cS[\lambda],D^{loc})& \cong H^p_d(\cA[\lambda],D^{loc})\oplus H^p_d(\cA[\lambda]\zeta,D^{loc}) \\ \notag 
                                 &  \cong H^p_d(\cA[\lambda],D^{loc})\oplus H^{p-1}_d(\cA[\lambda],D^{loc}),
    \end{align}
    which implies all desired vanishings (for $p\leq d$ the only non-trivial
    cohomology groups are
    $H^{0}_0(\cA[\lambda],D^{loc})\cong \mathbb{R}[\lambda]$ and
    $H^{3}_3(\cA[\lambda],D^{loc})\cong \bigoplus_{i=1}^N C^\infty(\mathbb{R},
    u^i)$~\cite[Theorems 2.12 and 2.13]{CKS}).

    Since both $H^{i}_i(\cS[\lambda],D^{loc})=0$ for $i=1,2,4$, and the induced
    differential on the first page has the $(p,d)$-degree $(1,1)$, we conclude
    that
    \begin{align}
      & H^0_0(\cS[\lambda], D_2-\lambda D_1) \cong H^0_0(\cS[\lambda],D^{loc}) \cong 	\mathbb{R}[\lambda]; \\
      & H^3_3(\cS[\lambda], D_2-\lambda D_1) \cong H^3_3(\cS[\lambda],D^{loc}) \cong 	\bigoplus_{i=1}^N C^\infty(\mathbb{R}, u^i).
    \end{align}
  \end{proof}

  \begin{remark} Almost the same statement holds for the cohomology of
    $\cS[\lambda]/\mathbb{R}$, the only difference is
    $H^0_0(\cS[\lambda]/\mathbb{R}, D_2-\lambda D_1) \cong 0$.
  \end{remark}

  Now we can complete the computation of the cohomology
  $H^p_d(\cE[\lambda], d_2-\lambda d_1)$ for $p<d$ and $p=2,d=2$ using the long
  exact sequence~\eqref{eq:long-exact}. The relevant pieces of this long exact
  sequence are
  \begin{align}
    \xymatrix{
    0=
    H^{p}_d(\cS[\lambda], D_2-\lambda D_1) \ar[r] &
                                                    H^{p}_d(\cE[\lambda], d^{loc}_2-\lambda d^{loc}_1) 
                                                    \ar@(d,u)[dl] \\
    H^{p+1}_{d}(\cS[\lambda]/\mathbb{R}, D_2-\lambda D_1) =0 & 
                                                               },
  \end{align}
  for $p<d$ and $(p+1,d)\not=(3,3)$, which implies the vanishing for $p<d$,
  $(p,d)\not=(2,3)$, and $p=2,d=2$. Moreover, we have
  \begin{align}
    \xymatrix{
    0=H^{2}_3(\cS[\lambda], D_2-\lambda D_1) \ar[r] &
                                                      H^{2}_3(\cE[\lambda], d^{loc}_2-\lambda d^{loc}_1) 
                                                      \ar@(d,u)[dl] \\
    H^{3}_{3}(\cS[\lambda]/\mathbb{R}, D_2-\lambda D_1)\cong 	\bigoplus_{i=1}^N C^\infty(\mathbb{R}, u^i) \ar[r] & H^{3}_{4}(\cS[\lambda], D_2-\lambda D_1) =0 &
                                                                                                                                                                   },
  \end{align}
  which gives the answer for $(p,d)=(2,3)$. Now, the special cases of these
  computations for $p=2$, $d\geq 2$ and $p=3$, $d\geq 4$ imply all statements
  of Theorem~\ref{thm:BiHomCoh}.
\end{proof}


An immediate corollary of Theorem~\ref{thm:BiHomCoh} is the following:

\begin{corollary}\label{cor:DefSSPencil} Let $\int P_2-\lambda P_1$ be a
  semi-simple pencil of local Poisson bivectors of differential order $1$. We
  consider the higher order dispersive extensions of $\int P_2-\lambda P_1$ in
  the realm of weakly non-local Poisson pencils of localizable shape, that is,
  we consider Poisson pencils
  $\int \sum_{d=1}^\infty \epsilon^{d-1}(P_{2,d}-\lambda P_{1,d})\in\cE$ such
  that $\deg_{\D_x}( P_{2,d}-\lambda P_{1,d}) = d$ and
  $\int P_{2,1}-\lambda P_{1,1} = \int P_2-\lambda P_1$.
	
  The space of orbits of the action of the group $\cR_{II}$ (the group of
  Miura-reciprocal transformation of the 2nd kind) onto the set of these
  dispersive extensions is isomorphic to the space
  $\bigoplus_{i=1}^N C^\infty(\mathbb{R},u^i)$.
\end{corollary}

This result is strikingly similar to the corresponding statement in the local
case, cf.~\cite[Theorem 1]{CarletPosthumaShadrin:DeformSS}, see
also~\cite{LiuZhang-deformss,LiuZhang:biham,DubrovinLiuZhang:QuasiTriv}.
However, in the local case both the space $\hat \cF$ where the deformations of
$\int P_2-\lambda P_1$ are allowed as well as the group $\cG_{II}$ acting on
them are much smaller than in Corollary~\ref{cor:DefSSPencil}. Our next goal is
to compare these two situations.

\subsection{Comparison with the purely local deformations}

Within this section it is important to have a notation that distinguishes
between the operator $\D_x$ as given by Equation~\eqref{eq:def-Partial_x-NL} on
the space $\cS = \hat\cA[\zeta]$ and its purely local version
$\tilde{\D}_x \coloneqq \D_x + u^{i,1}\theta_i \D_\zeta$ defined both on $\cS$
and on $\cA$. Note that on $\cS$ the operator $\tilde{\D}_x$ commutes with
multiplication by $\zeta$.

Let $T^{nl}$ denote the space of dispersive weakly non-local Poisson pencils of
localizable shape
$\int \sum_{d=1}^\infty \epsilon^{d-1}(P_{2,d}-\lambda P_{1,d})\in\cE$ with the
fixed leading term $\int P_{2,1}-\lambda P_{1,1} = \int P_{2}-\lambda P_{1}$
that is purely local and semi-simple.  Let $T^{loc}$ denote the space
of dispersive local Poisson pencils
$\int \sum_{d=1}^\infty \epsilon^{d-1}(P_{2,d}-\lambda P_{1,d})\in\hat\cF$ with
the same fixed leading term
$\int P_{2,1}-\lambda P_{1,1} = \int P_{2}-\lambda P_{1}$.

The group $\cR_{II}$ acts on $T^{nl}$ and the group $\cG_{II}$ acts on
$T^{loc}$. Moreover, there is a natural embedding $I \colon T^{loc} \to T^{nl}$
that is $\cG_{II}$-equivariant ($\cG_{II}$ acts on $T^{nl}$ as a subgroup of
$\cR_{II}$). The map $I$ induces a map of the sets of orbits
$\iota\colon T^{loc}/\cG_{II} \to T^{nl}/\cR_{II}$.

\begin{proposition} The map $\iota$ is injective.
\end{proposition}

\begin{proof} This proposition immediately follows from~\cite[Theorem
  1.3]{LiuZhang:JacobiStructures} and~\cite[Theorem
  2]{CarletPosthumaShadrin:DeformSS}. By the latter result in the local case,
  we have an isomorphism of sets
  $\tilde{c}\colon T^{loc}/\cG_{II} \to \bigoplus_{i=1}^N C^\infty(\mathbb{R},
  u^i)$ (these are the so-called central invariants in the local case). On the
  other hand, \cite[Theorem 1.3]{LiuZhang:JacobiStructures} states that for any
  $x,y\in T^{loc}/\cG_{II}$ such that $\iota(x)=\iota(y)$ we have
  $\tilde{c}(x)=\tilde{c}(y)$. Hence, $x=y$, and $\iota$ is surjective.
\end{proof}

Corollary~\ref{cor:DefSSPencil} implies that there is a $\cR_{II}$ invariant
map $C\colon T^{nl} \to \bigoplus_{i=1}^N C^\infty(\mathbb{R}, u^i)$ that
descends to a bijection
$c\colon T^{nl}/\cR_{II} \to \bigoplus_{i=1}^N C^\infty(\mathbb{R}, u^i)$. We
have the following

\begin{proposition} \label{prop:ciota} The composition
  $c\circ \iota \colon T^{loc}/\cG_{II} \to \bigoplus_{i=1}^N
  C^\infty(\mathbb{R}, u^i)$ is surjective.
\end{proposition}

\begin{proof} Basically, we want to show that any cohomology class in
  $H^2_3(\cE)$ has a representative with a purely local density. Let
  $\int a^2+a^1\zeta$ represent a class in $H^2_3(\cE)$,
  $a^2\in \hat\cA^2_3[\lambda]$ and $a^1\in \hat\cA^1_3[\lambda]$. This means
  that
  \begin{align}
    (D_2-\lambda D_1) (a^2+a^1\zeta )\in \D_x(S^3_3[\lambda]),
  \end{align}	
  or, in other words, that there exist $b^3\in \hat\cA^3_3[\lambda]$ and
  $b^2\in \hat\cA^2_3[\lambda]$ such that
  \begin{align}
    D^{loc}(a^2)-(\hat{E}(P_2)-\lambda \hat{E}(P_1)) (a^1) + D^{loc}(a^1)\zeta = \tilde{\D_x}(b^3) + u^{i,1}\theta_i b^2 + \tilde{\D_x} (b^2)\zeta. 
  \end{align}
  Since $H^1_3(\hat{\cA}[\lambda], D^{loc})=0$~\cite[Theorem 2.13]{CKS}, there
  exist $e^0\in \hat{\cA}^0_2[\lambda]$ and $f^1\in \hat{\cA}^1_2[\lambda]$
  such that $D^{loc} e^0 = a^1 + \tilde{\D_x}(f^1)$.  Then,
  \begin{align}
    (D_2-\lambda D_1) (e^0\zeta) & = a^1\zeta + \tilde{\D_x}(f^1)\zeta -(\hat{E}(P_2)-\lambda \hat{E}(P_1)) (e^0)
    \\ \notag 
                                 & = a^1\zeta -(\hat{E}(P_2)-\lambda \hat{E}(P_1)) (e^0) -u^{i,1}\theta_i f^1+\D_x(f^1\zeta),
  \end{align}
  which implies that
  \begin{align}
    (d_2-\lambda d_1) \int e^0\zeta = \int a^1\zeta -(\hat{E}(P_2)-\lambda \hat{E}(P_1)) (e^0) -u^{i,1}\theta_i f^1
  \end{align}
  Thus, the cocycle $\int a^2+a^1\zeta$ is cohomologous to
  $\int a^2+(\hat{E}(P_2)-\lambda \hat{E}(P_1)) (e^0) + u^{i,1}\theta_i f^1$,
  which gives a pure local deformation for $\int P_2-\lambda P_1$.
\end{proof}

Taking into account that that $c$ is a bijection, an immediate corollary of
Proposition~\ref{prop:ciota} is the following:

\begin{corollary} The map $\iota$ is surjective (and hence a bijection). In
  particular, every orbit of the action of $\cR_{II}$ on $T^{nl}$ contains a
  purely local representative.
\end{corollary}

It is just a different way to state
Theorem~\ref{thm:LocalLeadingTermOfAPencil}, so this corollary also completes
the proof of Theorem~\ref{thm:LocalLeadingTermOfAPencil}

\subsection{Roots of the characteristic polynomial of the symbol} In the purely
local case the central invariants, besides a purely cohomological definition,
can be computed directly from a representative of a deformation
(see~\cite{DubrovinLiuZhang:QuasiTriv} for details).  More precisely, one has
to compute the eigenvalues of the symbol of a representative of a deformation,
which behave as scalars with respect to the Miura group action.
In this section we extend this viewpoint to the invariants of the
Miura-reciprocal group.

First, we recall the construction from~\cite{DubrovinLiuZhang:QuasiTriv}. Let
$\int \sum_{d=1}^\infty \epsilon^{d-1}(P_{2,d}-\lambda P_{1,d})\in T^{loc}$,
and the densities are expanded as
$\sum_{s=0}^d (P_{2,d,s}-\lambda P_{1,d,s})^{ij}\theta_{i}\theta_{j}^{d-s}$,
$d\geq 1$, such that
\begin{align}
  \sum_{s=0}^d (P_{2,d,s}-\lambda P_{1,d,s})^{ij} \D_x^{d-s} = -\sum_{s=0}^d (-\D_x)^{d-s} \circ (P_{2,d,s}-\lambda P_{1,d,s})^{ji}
\end{align}
Consider the symbol of the densities of the bi-vector
$\int \sum_{d=1}^\infty \epsilon^{d-1}(P_{2,d}-\lambda P_{1,d})$, that is, the
sum
$\sum_{d=1}^\infty \epsilon^{d-1}(P_{2,d,0}-\lambda P_{1,d,0})^{ij}=
\sum_{d=1}^\infty (-\epsilon)^{d-1}(P_{2,d,0}-\lambda P_{1,d,0})^{ji}$. The
construction of the Miura group invariants from the eigenvalues of the symbol
is based on the following lemma:

\begin{lemma} Under the group of Miura transformations $\cG$ the symbol
  transforms linearly as a pencil of bi-linear forms:
  \begin{align}
    \sum_{d=1}^\infty \epsilon^{d-1}(P_{2,d,0}-\lambda P_{1,d,0})^{ij} \mapsto 
    \sum_{d=0}^\infty \epsilon^{d}\frac{\D w_d^i}{\D u^{k,d}}
    \sum_{d=1}^\infty \epsilon^{d-1}(P_{2,d,0}-\lambda P_{1,d,0})^{k\ell}
    \sum_{d=1}^\infty (-\epsilon)^{d}\frac{\D w_d^j}{\D u^{\ell,d}}
  \end{align}
  (here $w^i=\sum_{d=0}^\infty \epsilon^d w^i_d$, $w^i_d\in \cA_d$,
  $i=1,\dots,N$, are the new coordinates).
  Hence, the eigenvalues of this pencil behave as scalar with respect to the
  action of the Miura group.
\end{lemma}

There are $N$ roots $\lambda_i$, $i=1,\dots,N$, of the $\lambda$-polynomial
\begin{equation}\label{chpol}
  \mathrm{det}\left(\sum_{d=1}^\infty \epsilon^{d-1}(P_{2,d,0}-\lambda P_{1,d,0})\right)
\end{equation}
which are the formal power series in $\epsilon$ with the coefficients given by
smooth functions in $u^1,\dots,u^N$, with the leading terms in $\epsilon$ given
by $m_i = u^i + O(\epsilon)$. These eigenvalues are further used to derive the
closed formulas for the central invariants of a pencil
$\int \sum_{d=1}^\infty \epsilon^{d-1}(P_{2,d}-\lambda P_{1,d})\in T^{loc}$.

In the weakly non-local case of localizable shape, the densities of
$\int \sum_{d=1}^\infty \epsilon^{d-1}(P_{2,d}-\lambda P_{1,d})\in T^{nl}$ can
be uniquely expanded as
$\sum_{s=0}^d (P_{2,d,s}-\lambda P_{1,d,s})^{ij}\theta_{i}\theta_{j}^{d-s} +
(Q_{2,d}-\lambda Q_{1,d})^i \theta_i\zeta $, $d\geq 1$, such that
\begin{align}
  \sum_{s=0}^d (P_{2,d,s}-\lambda P_{1,d,s})^{ij} \D_x^{d-s} = -\sum_{s=0}^d (-\D_x)^{d-s} \circ (P_{2,d,s}-\lambda P_{1,d,s})^{ji}.
\end{align}
(this expansion we call the ``normal form'' below).

\begin{proposition} Let
  $\int \sum_{d=1}^\infty \epsilon^{d-1}(P_{2,d}-\lambda P_{1,d})\in T^{nl}$
  and let
  \[\lambda^i=r^i+\epsilon^2\lambda^i_2+ \epsilon^4\lambda^i_4+\cdots\]
  be the $\lambda$-roots of the characteristic polynomial \eqref{chpol}.
  The quantities
  \[c^i_{2k}=\frac{\lambda^i_{2k}}{(f^i)^k},\qquad k=0,1,2,...\] (where $f^i$
  are the diagonal entries of the first metric in canonical coordinates) are
  invariant under the action of $\cR$.
\end{proposition}

\begin{proof}
  Taking into account the above lemma we focus on pure reciprocal
  transformations.  The action of reciprocal transformations of \emph{1st kind}
  on the coefficients of the symbols can be easily obtained using the same
  arguments used in the proof of Proposition \eqref{prop:F-P}. Indeed
  \begin{eqnarray*}
    \tilde P^{ij}_\lambda&=&B^{-1} \left(B\delta^i_k - u^i_x \D_x^{-1} \frac{\D B}{\D u^k}\right) P^{k\ell}_\lambda 
                             \left(B\delta^j_\ell + \frac{\D B}{\D u^\ell} \D_x^{-1} u^j_x  \right)\\
                         &=&BP^{ij}_\lambda+P^{il}_\lambda\frac{\D B}{\D u^\ell} \D_x^{-1} u^j_x -\frac{1}{B}u^i_x \D_x^{-1} \frac{\D B}{\D u^k}P^{kj}_\lambda- u^i_x \D_x^{-1} \frac{\D B}{\D u^k} P^{k\ell}_\lambda\frac{\D B}{\D u^\ell} \D_x^{-1} u^j_x .
  \end{eqnarray*}
  The second, the third and the fourth terms cannot contribute to the symbol of
  $\tilde P^{ij}_\lambda$, while in the first term the only contributions come
  from
  \[B\sum_{d=1}^\infty \epsilon^{d-1}(P_{2,d,0}-\lambda
    P_{1,d,0})\partial_x^d=B\sum_{d=1}^\infty \epsilon^{d-1}(P_{2,d,0}-\lambda
    P_{1,d,0})B^d\partial_y^d=B^2\sum_{d=1}^\infty
    (B\epsilon)^{d-1}(P_{2,d,0}-\lambda P_{1,d,0})\partial_y^d\] that implies
  \[\sum_{d=1}^\infty \epsilon^{d-1}(P_{2,d,0}-\lambda P_{1,d,0})^{ij}\to
    B^2\sum_{d=1}^\infty(B\epsilon)^{d-1}(P_{2,d,0}-\lambda P_{1,d,0})^{ij}.\]
  This means that
  \[\lambda^i\to r^i+(B\epsilon)^2\lambda^i_2+(B\epsilon)^4\lambda^i_4+\cdots\]
  or, equivalently, that
  \[\lambda^i_{2k}\to B^{2k}\lambda^i_{2k}.\]
  The result then follows from the transformation rule for the contravariant
  metric (see \eqref{eq:F-P-loc}): $f^i\to B^2f^i$. In the case of reciprocal
  transformation of \emph{2nd kind} we observe that they do not affect the
  symbol of the pencil. Indeed a bivector transforms according to the following
  rule
  \begin{align}
    \tilde P^{ij}\coloneqq B^{-1} \left(B\delta^i_k - u^i_x \D_x^{-1} \frac{\D B}{\D u^k_{\sigma}}\partial_x^{\sigma}\right) P^{k\ell} 
    \left(B\delta^j_\ell + (-\partial_x)^{\tau}\frac{\D B}{\D u^\ell_\tau} \D_x^{-1} u^j_x  \right).
  \end{align}
  where
  \[B= 1+H=1 +
    \sum^\infty_{k=1}\epsilon^kH_k(u^j,u^j_x,\ldots,u^j_\sigma),\qquad
    H_k\in\mathcal{A}_k.\] Thus we have
  \begin{align}
    \tilde P^{ij}\coloneqq P^{ij}-\frac{1}{B}u^i_x \D_x^{-1} \frac{\D H}{\D u^k_{\sigma}}\partial_x^{\sigma}P^{kj}+\left(\delta^i_k - \frac{1}{B}u^i_x \D_x^{-1} \frac{\D H}{\D u^k_{\sigma}}\partial_x^{\sigma}\right) P^{k\ell} 
    \left(H\delta^j_\ell + (-\partial_x)^{\tau}\frac{\D H}{\D u^\ell_\tau} \D_x^{-1} u^j_x  \right).
  \end{align}
  Since the symbol of the bivector contains only the subset of the coefficients
  which depend only on the $u$'s but not on their $x$-derivatives the second
  term and the third terms above cannot contribute to it. This implies that the
  symbol of each bivector defining the pencil is unaffected by these
  transformations. For Miura reciprocal transformations \eqref{eq:756} the
  transformation rule for the symbol of the pencil is obtained combining the
  Lemma 4.11 with the above rule. It turns out that the symbol of the pencil
  transforms in the following way
  \begin{align}
    \sum_{d=1}^\infty \epsilon^{d-1}(P_{2,d,0}-\lambda P_{1,d,0})^{ij} \mapsto 
    B^2\sum_{d=0}^\infty \epsilon^{d}\frac{\D w_d^i}{\D u^{k,d}}
    \sum_{d=1}^\infty (B\epsilon)^{d-1}(P_{2,d,0}-\lambda P_{1,d,0})^{k\ell}
    \sum_{d=1}^\infty (-\epsilon)^{d}\frac{\D w_d^j}{\D u^{\ell,d}}.
  \end{align}

\end{proof}


\section{Projective-reciprocal invariance of the Doyle--Pot\"emin form}

\label{sec:Projective}

In this Section we make a first step towards the study of the
projective-reciprocal group action.  Consider a local operator of homogeneous
differential order $d+2$, $d\geq 2$ of the form
$P^{ij}=\D_x\circ Q^{ij}\circ\D_x$. We call this presentation of an operator
the \emph{Doyle--Pot\"emin form} (see Subsection~\ref{sec:proj-group-doyle}).

We prove the following theorem:

\begin{theorem}
  The projective group preserves the Doyle--Pot\"emin form of an operator. More
  precisely, the image of a homogeneous skew-symmetric operator of the form
  $\D_x \circ Q^{ij} \circ \D_x$, $\deg_{\D_x}Q^{ij} = d\geq 0$ under the
  action of an element of $\cP$ is a homogeneous skew-symmetric operator of
  the form $\D_x \circ \tilde Q^{ij} \circ \D_x$,
  $\deg_{\D_x}\tilde Q^{ij} = d\geq 0$
\end{theorem}

\begin{proof} Consider an element of the group $\cP$ given by 	
\begin{align}
		dy & = A^0dx,\\ \notag 
		w^i & =A^i/A^0, &&  i=1,\ldots,N,
\end{align}
where $A^i\coloneqq a^i_j u^j + a^i_0$, $i=0,1,\dots,N$.  Since the functions
$A^i$ and $A^0$ do not depend on the higher jet variables,
Theorem~\ref{thm:Miura-Reciprocal-Bivector} implies that the operator
$P^{ij}=\D_x \circ Q^{ij} \circ \D_x$ in the coordinates $y,w^1,\dots,w^N$ is
represented as
\begin{align}
	\tilde{P}^{ij} & = \frac 1{A^0} \Bigg(A^0 \pdS{\Big(\frac{A^i}{A^0}\Big)}{u^k} -\D_x\Big(\frac{A^i}{A^0}\Big) \D_x^{-1} \circ \pdS{A^0}{u^k}\Bigg)\D_x \circ Q^{kl} \circ \D_x \circ 
	\\ \notag & \qquad 
	\Bigg( \pdS{\Big(\frac{A^j}{A^0}\Big)}{u^l} A^0 + \pdS{A^0}{u^l} \D_x^{-1} \circ\D_x\Big(\frac{A^j}{A^0}\Big) \Bigg)
\end{align}
Now we see that
\begin{align}
	&  \D_x \circ \Bigg( \pdS{\Big(\frac{A^j}{A^0}\Big)}{u^l} A^0 + \pdS{A^0}{u^l} \D_x^{-1} \circ\D_x\Big(\frac{A^j}{A^0}\Big) \Bigg)
	 = \D_x \circ \Bigg( a^j_l - a^0_l \Big(\frac{A^j}{A^0}\Big) + a^0_l \D_x^{-1} \circ\D_x\Big(\frac{A^j}{A^0}\Big) \Bigg)
	\\ \notag & 
	=  \Bigg( a^j_l - a^0_l \Big(\frac{A^j}{A^0}\Big) \Bigg) \D_x 
	= (A^0)^2 \pdS{w^j}{u^l}  \D_y
\end{align}
and analogously
\begin{align}
	& \frac 1{A^0} \Bigg(A^0 \pdS{\Big(\frac{A^i}{A^0}\Big)}{u^k} -\D_x\Big(\frac{A^i}{A^0}\Big) \D_x^{-1} \circ \pdS{A^0}{u^k}\Bigg)\D_x
	= \frac 1{A^0} \Bigg(a^i_k -\Big(\frac{A^i}{A^0}\Big)a^0_k -\D_x\Big(\frac{A^i}{A^0}\Big) \D_x^{-1} \circ a^0_k\Bigg)\D_x
	\\ \notag &
	= \frac 1{A^0} \D_x \circ \Bigg(a^i_k -\Big(\frac{A^i}{A^0}\Big)a^0_k \Bigg) =
	\D_y \circ \frac{1}{A^0} \pdS{w^i}{u^k} (A^0)^2. 
\end{align}
Thus we see that $\tilde P^{ij}$ takes the form $\D_y \circ \tilde Q^{ij} \circ \D_y$, where the operator $ \tilde Q^{ij}$ is equal to 
\begin{align}
\tilde Q^{ij} = \frac 1{A^0} \pdS{w^i}{u^k} (A^0)^2 Q^{kl} (A^0)^2 \pdS{w^j}{u^l} 	
\end{align}
after the substitution $w^i(u^1,\dots,u^N)=A^i/A^0$ and $\D_y = (A^0)^{-1}\D_x$, which makes it manifestly skew-symmetric and homogeneous of the same degree.
\end{proof}

\begin{remark}
  Note that we don't use the Poisson property in the proof (and we don't have
  it in the statement of the theorem). This allows us to apply the
  projective-reciprocal transformation to any homogeneous skew-symmetric
  operators of the Doyle--Pot\"emin form, and the action would preserve the
  form.
\end{remark}

\begin{remark}
  Interesting examples of skew-symmetric operators in the Doyle--Pot\"emin form
  are coming from the theory of Dubrovin--Zhang
  hierarchies~\cite{dubrovin01:_normal_pdes_froben_gromov_witten}. Is it proved
  in~\cite{BuryakPosthumaShadrin:PolynomialBracket,BPS-2} that Dubrovin--Zhang
  hierarchies posses a Poisson bracket given by an operator of the shape
  $\sum_{p=0}^\infty \epsilon^{2p} P^{ij}_{2p+1}$, where
  $P^{ij}_1 = \eta^{ij}\D_x$ for some constant inner product $\eta^{ij}$, and
  for $p\geq 1$ the operators $P^{ij}_{2p+1}$ are homogeneous skew-symmetric
  operators of the shape $\sum_{e=0}^{2p+1} P^{ij}_{2p+1,e}\D_x^{2p+1-e}$,
  where $\deg_{\D_x} P^{ij}_{2p+1,e} = e$, such that $P^{ij}_{2p+1,0} =
  0$. Using that the operators $P^{ij}_{2p+1}$, $p\geq 1$, are skew-symmetric,
  it is easy to show that each of them is of the Doyle--Pot\"emin form.
\end{remark}

\printbibliography

\end{document}